\documentclass[sigconf]{acmart}

\usepackage{subcaption}
\usepackage{multirow}
\usepackage{xcolor}
\usepackage{multicol}
\usepackage{pdfcomment}
\usepackage[capitalize]{cleveref}
\usepackage[disable]{todonotes}
\setuptodonotes{size=\tiny, inline}

\usepackage{algorithm}
\usepackage[noend]{algpseudocode}
\algrenewcommand\algorithmicindent{1.0em}

\algnewcommand{\IfThen}[2]{
  \State \algorithmicif\ #1\ \algorithmicthen\ #2}
\algnewcommand{\IfThenElse}[3]{
  \State \algorithmicif\ #1\ \algorithmicthen\ #2\ \algorithmicelse\ #3}

\newcommand{\bhline}[1]{\noalign{\hrule height #1}}

\newcommand{\lcolon}{{:\,}}

\newcommand{\revisiontag}[3]{}
\newcommand{\revisiontagx}[3]{}
\newcommand{\revisiontaghere}[3]{}
\newcommand{\revisiontext}[1]{#1}
\newcommand{\revision}[4]{\revisiontext{\revisiontag{#1}{#2}{#3}#4}}
\newcommand{\revisionx}[4]{\revisiontext{\revisiontagx{#1}{#2}{#3}#4}}

\defineavatar{R1}{color=red}
\defineavatar{R2}{color=yellow}
\defineavatar{R3}{color=green}

\setcopyright{acmcopyright}
\copyrightyear{2018}
\acmYear{2018}
\acmDOI{XXXXXXX.XXXXXXX}

\acmConference[Conference acronym 'XX]{Make sure to enter the correct
  conference title from your rights confirmation email}{June 03--05,
  2018}{Woodstock, NY}
\acmPrice{15.00}
\acmISBN{978-1-4503-XXXX-X/18/06}

\begin{document}

\title{GuP: Fast Subgraph Matching by Guard-based Pruning}

\author{Junya Arai}
\email{junya.arai@ntt.com}
\orcid{0000-0002-5681-5941}
\affiliation{%
  \institution{Nippon Telegraph and Telephone Corporation}
  \streetaddress{3-9-11 Midoricho}
  \city{Musashino-shi}
  \state{Tokyo}
  \country{Japan}
  \postcode{180-8585}
}

\author{Yasuhiro Fujiwara}
\email{yasuhiro.fujiwara@ntt.com}
\affiliation{%
  \institution{Nippon Telegraph and Telephone Corporation}
  \streetaddress{3-1 Morinosato Wakamiya}
  \city{Atsugi-shi}
  \state{Kanagawa}
  \country{Japan}
  \postcode{243-0198}
}

\author{Makoto Onizuka}
\email{onizuka@ist.osaka-u.ac.jp}
\affiliation{%
  \institution{Osaka University}
  \streetaddress{1-5 Yamadaoka}
  \city{Suita-shi}
  \state{Osaka}
  \country{Japan}
  \postcode{565-0871}
}

\begin{abstract}
Subgraph matching, which finds subgraphs isomorphic to a query, is the key to information retrieval from data represented as a graph.
To avoid redundant exploration in the data, existing methods restrict the search space by extracting candidate vertices and candidate edges that may constitute isomorphic subgraphs.
However, it still requires expensive computation because candidate vertices induce many subgraphs that are not isomorphic to the query.
In this paper, we propose GuP, a subgraph matching algorithm with pruning based on guards.
Guards are a pattern of intermediate search states that never find isomorphic subgraphs.
GuP attaches a guard on each candidate vertex and edge and filters out them adaptively to the search state.
The experimental results showed that GuP can efficiently solve various queries, including those that the state-of-the-art methods could not solve in practical time.
\end{abstract}



\settopmatter{printfolios=true}

\maketitle

\section{Introduction}
\label{sec:introduction}

Similar to searching for a specific phrase within a document, searching for a specific structure within a graph is one of the most fundamental operations in graph databases.
This operation is formally defined as \emph{subgraph matching}.
  It enumerates all full embeddings, which map every vertex in a \emph{query graph} to the vertex of an isomorphic subgraph in a \emph{data graph}.
  We refer to a vertex in the query graph and the data graph as a query vertex and a data vertex, respectively.
\cref{fig:overall-example} shows an example of query graph $Q$ and data graph $G$.
Letting $(u_i, v_j)$ denote an assignment of query vertex $u_i$ to data vertex $v_j$, there exists full embedding $M = \{ (u_0, v_1), (u_1, v_4), (u_2, v_7), (u_3, v_{10}), (u_4, v_0) \}$.
Since subgraph matching is NP-hard \cite{Han2013a} and computationally expensive for complex graphs, efficient methods have been studied for a long time \cite{Ullmann1976, Cordella2004, Han2013a, Bi2016, Han2019, Sun2020tkde, Sun2020sigmod, Sun2021, Kim2022}.

Mainstream methods for subgraph matching perform a \emph{backtracking} search,
  which is a recursive procedure that extends \emph{partial embedding} $M$ with a new assignment of a query vertex and recurses with extended $M$ until $M$ becomes a full embedding.
  The extension is performed so that $M$ preserves isomorphism.
  If such an extension is impossible, the procedure returns to the caller, and the caller tries extending $M$ with another assignment.
A partial embedding is called a \emph{deadend} if it is not extendable or fails to yield any full embeddings in the subsequent recursions \cite{Rossi2006}.
Since it is futile to perform recursions with deadend partial embeddings, reducing them is the key to improving the search performance.

To reduce futile recursions, most methods employ candidate filtering \cite{Bi2016}.
  For each query vertex $u_i$, candidate filtering collects data vertices that can be a destination of $u_i$ into $C(u_i)$, a set of the candidate vertices of $u_i$.
    One of the most primitive filters is based on labels;
      it makes $C(u_i)$ of the data vertices with the same label as $u_i$.
    To further remove unnecessary vertices, modern methods perform matching with a tree or a directed acyclic graph (DAG) obtained from a query graph \cite{Han2013a, Bi2016, Bhattarai2019, Han2019, Kim2022}.
      They manage the candidate vertices and the edges between them, which we call candidate edges, in an auxiliary data structure such as a \emph{candidate space} \cite{Han2019}.
These approaches achieve a significant speedup of the search.

\begin{figure}[t]
  \begin{subfigure}[b]{0.30\linewidth}
    \centering
    \includegraphics[width=0.39\linewidth]{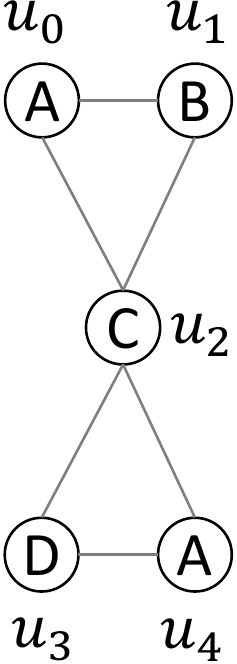}
    \vspace{0.5em}
    \caption{Query graph $Q$}
    \label{fig:overall-example-query}
  \end{subfigure}
  \begin{subfigure}[b]{0.69\linewidth}
    \centering
    \includegraphics[width=0.8\linewidth]{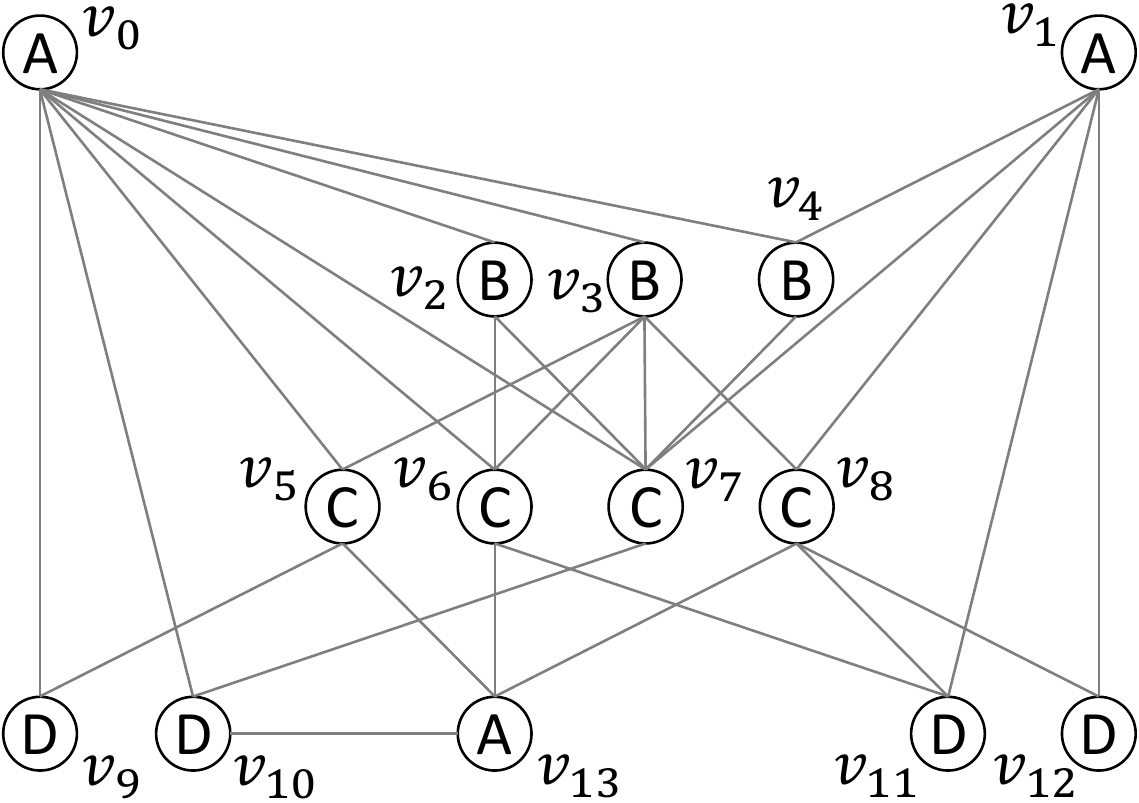}
    \caption{Data graph $G$}
    \label{fig:overall-example-graph}
  \end{subfigure}
  \vspace{-5mm}
  \caption{Example of a query graph and a data graph}
  \label{fig:overall-example}
\end{figure}

However, even if candidate filtering is applied, backtracking still suffers from numerous futile recursions.
This is because candidate filtering hardly captures a conflict between assignments, namely, a constraint violation caused by a combination of multiple assignments.
  Subgraph isomorphism requires that adjacent query vertices $u_i$ and $u_j$ are assigned to adjacent data vertices, constraining the combination of assignments of $u_i$ and $u_j$.
    This implies that a cycle in a query graph must be mapped to a cycle in a data graph.
    Cycles are usually difficult to find because of the sparseness of real-world graphs \cite{Chung2010}, and so partial embeddings tend to become deadends.
  Moreover, embeddings must be injective; namely, each query vertex must be assigned to a different data vertex.
    This globally constrains the combination of assignments in a partial embedding.
  These constraints are not well captured in the extraction of $C(u_i)$ because it is based on the constraints on only $u_i$ without assumptions on the assignments of the other query vertices.
Thus, candidate filtering fails to eliminate deadends due to conflicting assignments.

\paragraph{Our approach}
We propose \emph{GuP}, an efficient algorithm for subgraph matching.
In contrast to candidate filtering that captures constraints on a single vertex, GuP utilizes a \emph{guard} to capture constraints on a partial embedding.
  A guard is attached to each candidate vertex and candidate edge.
  If a partial embedding matches the attached guard, GuP adaptively filters out that vertex or edge.
  This enables early pruning of deadend partial embeddings before detecting a violation of the constraints.
In detail, GuP combines two kinds of guards: a \emph{reservation guard} and a \emph{nogood guard}.

\revision{R3}{O1}{novelty}{
The reservation guards propagate the injectivity constraint for checking it in earlier backtracking steps.
  Let us denote candidate vertex $v$ of query vertex $u_i$ by $(u_i, v)$, the same notation as an assignment.
  Intuitively, a reservation guard on $(u_i, v)$ is a set of the data vertices to be used in future extensions of partial embeddings with assignment $(u_i, v)$.
  The data vertices in the reservation guard must be kept unassigned in a partial embedding before extending it with $(u_i, v)$; otherwise, it violates the injectivity constraint in the subsequent extensions.
  Hence, we can filter out $v$ from $C(u_i)$ in such cases.

On the other hand, the nogood guards detect deadends by learning conflicting assignments from the deadends encountered before.
  A nogood guard is attached to both candidate vertex and candidate edge to exploit edges in a query graph for pruning.
  Conceptually, a nogood guard is a set of assignments that conflict with the extension using that candidate vertex or edge.
  Thus, we can filter out it if a partial embedding includes the assignments in the nogood guard.
  GuP updates nogood guards on-the-fly during backtracking by discovering a \emph{nogood} \cite{Stallman1977}, a set of conflicting assignments.
    Although a nogood is a widely-known concept for the constraint satisfaction problem, only several studies \cite{McCreesh2015, Mccreesh2020} applied it to subgraph matching.
    In this study, we introduce novel nogood discovery rules and a \emph{search-node encoding} for effective and efficient pruning.
    Our nogood discovery rules offer a general nogood, which can be found in many partial embeddings and hence offer high pruning power.
    GuP also has a special rule for a nogood guard on edges, while existing rules cannot produce a nogood that fits edge-based pruning.
    Furthermore, a search-node encoding provides a compact representation for a nogood guard and enables pruning without increasing the time and space complexities.

GuP stores guards in a \emph{guarded candidate space} (GCS), an auxiliary data structure with the guards.
The experimental results confirmed that guards significantly reduce futile recursions, and as a result, GuP can process query graphs that cannot be processed by the state-of-the-art methods even after spending an hour.
Our contributions introduced in this paper are summarized as follows:

\begin{enumerate}
  \item Pruning approach based on guards,
  \item Reservation, a pruning condition based on injectivity,
  \item Nogood discovery rules to obtain general nogoods, and
  \item Search-node encoding of a nogood guard.
\end{enumerate}
}

\paragraph{Paper organization}
\cref{sec:background} presents the background, and \cref{sec:method} details our approach.
\cref{sec:evaluation} discusses the experimental results, and we conclude this paper in \cref{sec:conclusion}.

\section{Background}
\label{sec:background}

In this section, we review related work and introduce the problem definition and notations used in this paper.

\subsection{Related Work}
\label{sec:related-work}

There are various problem settings and algorithms related to the search of subgraphs, such as subgraph enumeration algorithms for unlabeled graphs \cite{Lai2015, Lai2016, Kim2016} and RDF query engines \cite{Huang2011, Yuan2013, Kim2015} for edge-labeled graphs.
On vertex-labeled graphs, subgraph containment algorithms \cite{Cheng2007, Bonnici2010, Giugno2013} take a set of data graphs and find ones with at least one embedding of a query graph, and subgraph matching algorithms find all embeddings in a single data graph.
Approaches based on join operations \cite{Aberger2017, Mhedhbi2019, Sun2021} are mainly used for subgraph homomorphism-based subgraph matching, which allows duplicate assignments of query vertices to the same data vertex.
In contrast, subgraph isomorphism-based subgraph matching, the focus of this paper, prohibits it.
Since most algorithms for this problem setting perform a backtracking search \cite{Ullmann1976}, we review three popular approaches to improve the efficiency of backtracking.

\emph{Candidate filtering.}
Conventional filtering methods are based on local features.
Ullmann \cite{Ullmann1976} employed label-and-degree filtering (LDF), which collects data vertex $v$ as a candidate vertex of $u_i$ if $v$ has the same label as $u_i$ and $v$ has a degree greater than or equal to that of $u_i$.
Neighborhood label frequency filtering (NLF) \cite{Bi2016} checks for every label $l$ if a candidate vertex of $u_i$ has label-$l$ neighbors not fewer than those of $u_i$.
For example, $v_{13}$ in \cref{fig:overall-example} is removed from $C(u_0)$ because $v_{13}$ has no label-$B$ neighbor although $u_0$ has one label-$B$ neighbor, $u_1$.
Recent methods employ LDF and NLF in common, but they also perform pseudo-matching on nearby vertices of a candidate vertex and a query vertex \cite{He2008, Zhao2010, Sun2020tkde} or matching with a spanning tree or a DAG built from a query graph \cite{Han2013a, Bi2016, Bhattarai2019, Han2019, Kim2022}.
All the previous approaches extract a candidate-vertex set before backtracking and do not change it after that.
In contrast, GuP adaptively changes it depending on a partial embedding by using guards.

\emph{Optimization of matching order.}
The size of the search space varies depending on matching order, in which the destination of query vertices is determined.
This is because the destination of query vertex $u_i$ must be chosen from data vertices adjacent to the destinations of all the matched neighbors of $u_i$.
Many efforts have been made to generate a good matching order that  first decides the destinations of query vertices with fewer candidate vertices and keeps the search space of the remaining query vertices small \cite{Shang2008, He2008, Bi2016, Han2013a, Han2019, Sun2020tkde}.
However, we still do not have a method that can generate a good order for arbitrary query graph and data graph \cite{Katsarou2017,Sun2020sigmod}.
Thus, it is important to reduce the number of candidate vertices and edges.

\emph{Use of nogoods.}
A nogood was introduced by Stallman and Sussman in 1977 \cite{Stallman1977} and has been well studied in the AI community.
Pruning with nogoods is performed in two ways: \emph{backjumping} \cite{Stallman1977, Gaschnig1978, Prosser1993} and \emph{nogood recording} \cite{Stallman1977, Ginsberg1993, Jussien2000}.
  Backjumping abandons deadends by escaping from ongoing recursions until the assignment shared with the last discovered nogood is changed.
  On the other hand, nogood recording stores discovered nogoods in a database and prunes partial solutions including a recorded nogood.
\revisiontext{
\revisiontagx{R1}{D1}{relatedwork}
Backjumping was also independently proposed in the database community.
\revisiontagx{R3}{O1}{failingset-relatedwork}
DAF \cite{Han2019} performs failing set-based pruning \cite{Han2019}, and VEQ \cite{Kim2022} captures equivalences of vertices in backjumping, like symmetricity-based nogood discovery \cite{Freuder1995}.
GuP also performs backjumping, but unlike DAF and VEQ, pruning with nogood guards is categorized as a nogood recording method.
The combination of backjumping and nogood recording enables GuP to eliminate more search space.
}

\subsection{Definitions}
\label{sec:definitions}

\begin{table}[!t]
  \setlength\tabcolsep{1mm}
  \footnotesize
  \centering
  \caption{Notations}
  \vspace{-3mm}
  \label{tb:notations}
  \begin{tabular}{cl}
    \bhline{1pt}
    \multicolumn{1}{c}{Symbol} & \multicolumn{1}{c}{Definition} \\
    \hline
    $u_i, v$           & Query vertex and data vertex \\
    $N(v)$             & Neighbor set of a vertex \\
    $N_-(u_i)$         & Backward neighbor set: $\{ u_j \in N(u_i) \mid j < i \}$ \\
    $N_+(u_i)$         & Forward neighbor set: $\{ u_j \in N(u_i) \mid j > i \}$ \\
    $C(u_i)$           & Candidate-vertex set of $u_i$ \\
    $M(u_i)$           & Destination of $u_i$ under $M$ \\
    $M \oplus v$       & Extension of $M$ with $v$: $M \cup \{(u_{|M|}, v)\}$ \\
    $M[K]$             & Restriction of $M$ to $K$: $\{ (u_i, v) \in M \mid u_i \in K \}$ \\
    $M[\lcolon i]$     & Restriction of $M$ by ID filtering: $\{ (u_j, v) \in M \mid j < i \}$ \\
    $V_Q[\lcolon i]$   & Set of query vertices by ID filtering: $\{ u_j \in V_Q \mid j < i \}$ \\
    $\mathrm{dom}(M)$  & Domain of a mapping: $\{ u_i \mid \exists v, (u_i, v) \in M \}$ \\
    $\mathrm{Im}(M)$   & Image of a mapping: $\{ v \mid \exists u_i, (u_i, v) \in M \}$ \\
    $R$                & Reservation guard on a candidate vertex \\
    $\mathit{NV}$, $\mathit{NE}$ & Nogood guards on a candidate vertex and a candidate edge \\
    \hline
  \end{tabular}
\end{table}

\cref{tb:notations} lists the symbols used in this paper.
We focus on vertex-labeled simple undirected graphs, similarly to previous studies \cite{Ullmann1976,Cordella2004,Han2013a,Bi2016,Han2019,Sun2020tkde,Sun2020sigmod,Sun2021, Kim2022}.
Note that our method can easily adapt to other kinds of graphs, such as directed graphs and edge-labeled graphs.
Consider query graph $Q = (V_Q, E_Q, \Sigma, \ell)$ and data graph $G = (V_G, E_G, \Sigma, \ell)$.
$V_Q$ and $V_G$ are sets of vertices, $E_Q$ and $E_G$ are sets of edges, $\Sigma$ is a set of labels, and $\ell$ is a mapping of a vertex to its label.
We assume that the query vertices have consecutive ID numbers, i.e., $V_Q = \{ u_0, u_1, u_2, \ldots \}$.
If there is no ambiguity, we use $u_i$ as a query vertex and $v$ as a data vertex without explicit mention.
$N(u_i)$ and $N(v)$ denote the sets of neighbors of $u_i$ and $v$, respectively.
Let the set of forward neighbors $N_+(u_i) = \{ u_j \mid j > i \}$ and the set of backward neighbors $N_-(u_i) = \{ u_j \mid j < i \}$.
For arbitrary domain $X \subseteq V_Q$, mapping $M: X \to V_G$ is denoted by binary relation $M \subseteq X \times V_G$.
$\mathrm{dom}(M)$ and $\mathrm{Im}(M)$ denote the domain and the image of $M$, respectively.
Notation $M(u_i)$ implicitly implies assumption $u_i \in \mathrm{dom}(M)$.
Given query graph $Q$ and data graph $G$, an embedding is defined as follows.

\begin{definition}[Embedding]
  \label{def:embedding}
  Mapping $M: V_Q \to V_G$ is an \emph{embedding} of $Q$ into $G$ if and only if $M$ satisfies the following constraints:
  \begin{enumerate}
    \item \emph{Label constraint}: $\forall u_i \in V_Q,\, \ell(u_i) = \ell(M(u_i)),$
    \item \emph{Adjacency constraint}: $\forall (u_i, u_j) \in E_Q,\, (M(u_i), M(u_j)) \in E_G,$
    \item \emph{Injectivity constraint}: $\forall u_i, u_j \in V_Q,\, i \neq j \Rightarrow M(u_i) \neq M(u_j).$
  \end{enumerate}
\end{definition}

Then, the problem definition of this paper is given as follows.

\begin{definition}[Subgraph matching]
  Given query graph $Q$ and data graph $G$, \emph{subgraph matching} is a problem to enumerate all the embeddings of $Q$ in $G$.
\end{definition}

The three constraints in \cref{def:embedding} are referred to as the \emph{constraints of isomorphism}.
An embedding is also called a \emph{full embedding} to emphasize the contrast to a \emph{partial embedding}, which is an embedding of an induced subgraph of $Q$.
The \emph{length} of partial embedding $M$ is the number of assignments in $M$, denoted by $|M|$.
An \emph{extension} is a mapping made by extending a partial embedding with an additional assignment.
In contrast to partial embeddings, extensions may violate the constraints of isomorphism.

In the following discussion, we assume that the matching order is ascending order of query vertex IDs.
This preserves the generality because renumbering vertex IDs can change the matching order.
Additionally, we assume that vertex IDs are numbered in a connected order \cite{Sun2020tkde}, that is, every query vertex except $u_0$ has a neighbor with a smaller vertex ID.
Note that matching orders used in subgraph matching usually satisfy this property \cite{Lee2012a, Sun2020sigmod}.
Under our assumptions, partial embeddings and extensions of length $k$ always consist of assignments of $u_0, u_1, \ldots, u_{k - 1}$.
We denote an extension of partial embedding $M$ with an assignment to $v$ by $M \oplus v$ i.e., $M \oplus v = M \cup \{ (u_k, v) \}$ where $k$ is the length of $M$.

\section{Method}
\label{sec:method}

We propose GuP, an efficient algorithm for subgraph matching.
GuP prunes deadend partial embeddings by filtering out unnecessary candidate vertices and edges adaptively to the assignments in each partial embedding.
The key idea is guards attached to each candidate vertex and edge, which represent a filtering condition.
In the following sections, we first present the overview of GuP in \cref{sec:overview}.
Then, \cref{sec:reservation-guard,sec:nogood-guard} details a reservation guard and a nogood guard, respectively, and the backtracking algorithm using guards is described in \cref{sec:backtracking}.
In addition, we introduce a search-node encoding of nogood guards and briefly discuss an approach for parallelization in \cref{sec:optimizations}.
Finally, we analyze the complexity of GuP in \cref{sec:complexity-analysis}.
All the examples in this section consider query graph $Q$ and data graph $G$ shown in \cref{fig:overall-example}.

\subsection{Overview}
\label{sec:overview}

GuP consists of the following steps.
\begin{enumerate}
  \item \emph{Guarded candidate space (GCS) construction:}
    GuP builds a GCS, an auxiliary data structure that organizes candidate vertices, candidate edges, and guards.
    This step includes candidate filtering and matching order optimization.
  \item \emph{Reservation guard generation:}
    GuP populates the reservation guards in the GCS by analyzing the connections between candidate vertices.
  \item \emph{Backtracking search:}
    GuP enumerates full embeddings using the data in the GCS.
    Nogood guards are generated on-the-fly and used for pruning along with reservation guards.
\end{enumerate}
In the GCS generation step, GuP employs extended DAG-graph DP \cite{Kim2022} for candidate filtering and VC \cite{Sun2020tkde} for optimizing the matching order.
  In the following, we assume that the query vertices are numbered in the optimized order.
  Note that an approach for candidate filtering and matching order optimization is out of the scope of this work, and guard-based pruning can be used in combination with arbitrary existing approaches.
\cref{fig:overall-example-gcs} illustrates the GCS for $Q$ and $G$ shown in \cref{fig:overall-example}.
  $R$ and $NV$ shows a reservation guard and a nogood guard attached to each candidate vertices.
  Nogood guards are also attached to candidate edges while they are omitted for conciseness.
  All the candidate-vertex sets are simply a set of data vertices with the same label, except that $v_{13}$ is removed from $C(u_0)$ by NLF as described in \cref{sec:related-work}.
Similar to a candidate space, a GCS provides all the information necessary for backtracking \cite{Han2019}.
After the GCS construction, GuP generates a reservation guard and then starts backtracking.

\begin{figure}[t]
  \centering
  \includegraphics[width=1.0\linewidth]{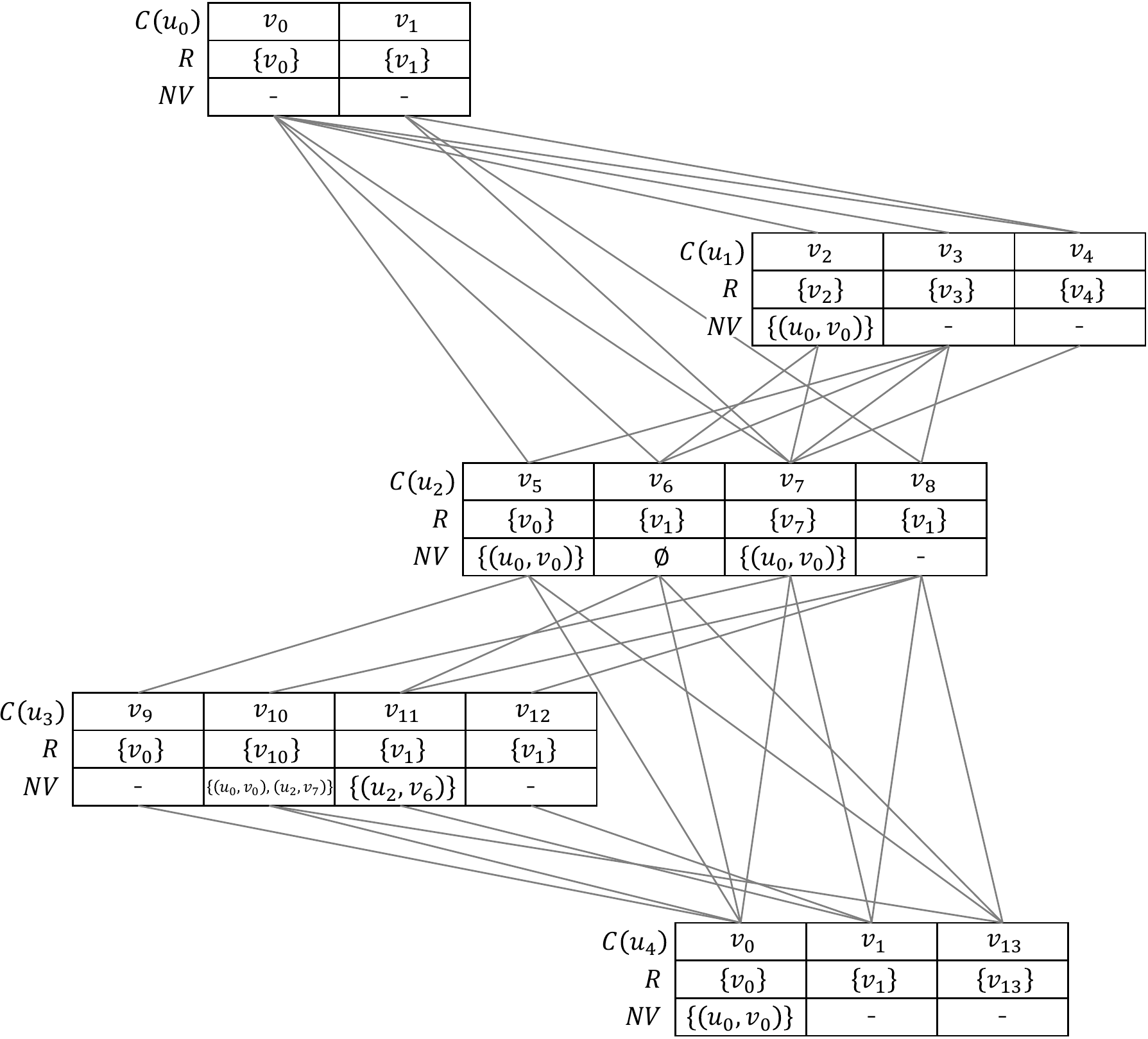}
  \vspace{-3mm}
  \caption{
    \revisiontext{Guarded candidate space for $Q$ and $G$ in \cref{fig:overall-example}.}
  }
  \label{fig:overall-example-gcs}
\end{figure}

\begin{figure}[t]
  \centering
  \includegraphics[width=0.9\linewidth]{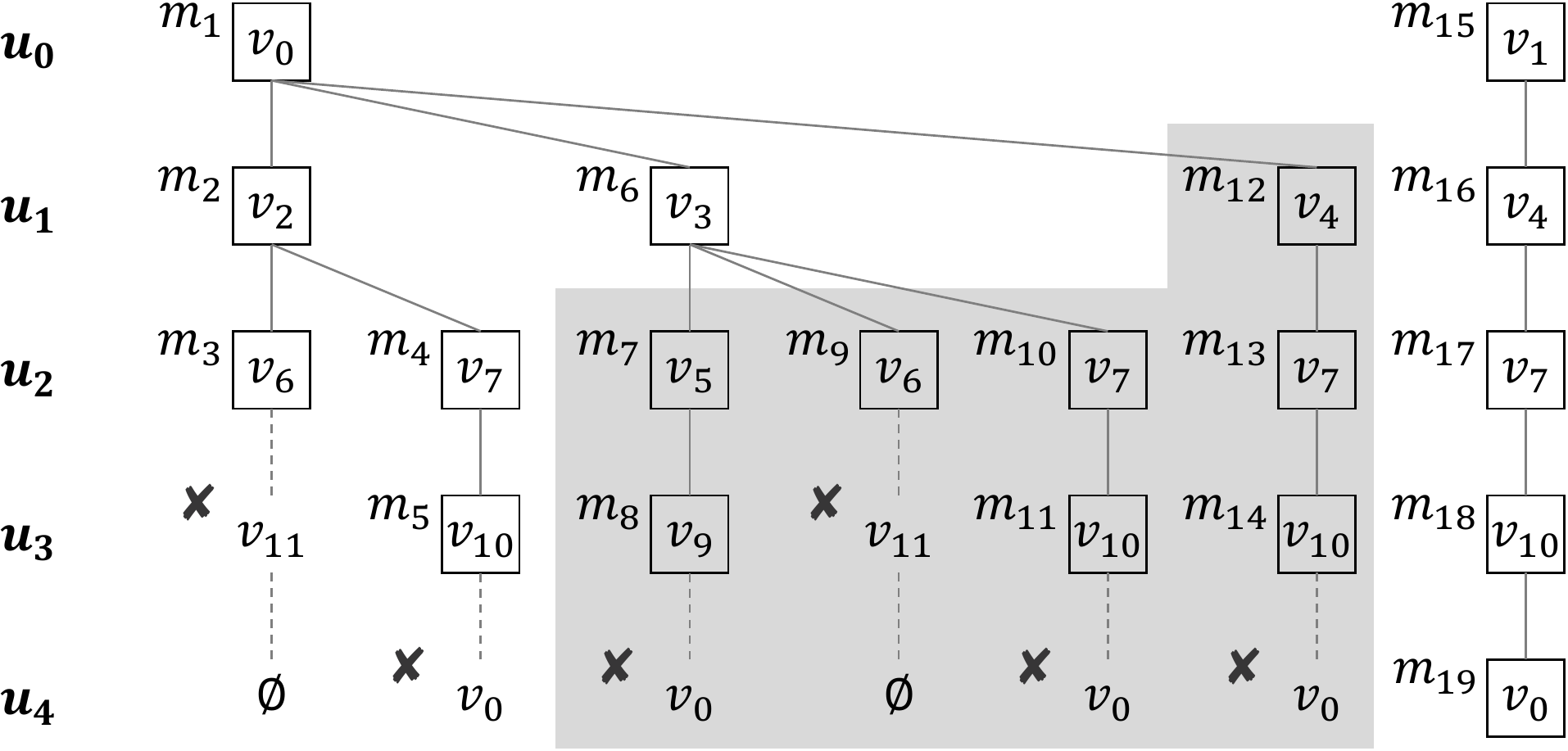}
  \vspace{-2mm}
  \caption{
    Search tree for $Q$ and $G$ in \cref{fig:overall-example}.
  }
  \label{fig:overall-example-search-tree}
\end{figure}

\subsection{Reservation guard}
\label{sec:reservation-guard}

This section first introduces an abstract concept of a reservation and then presents an algorithm to generate a reservation guard.

\subsubsection{Reservation}

We begin with the fundamental definitions.
Let $Q[I]$ be a subgraph of $Q$ induced by $I \subseteq V_Q$,

\begin{definition}[Inclusive descendant]
  \label{def:inclusive-descendant}
  Let $u_i, u_j \in V_Q$.
  $u_j$ is an \emph{inclusive descendant} of $u_i$ if and only if $u_i = u_j$ or $u_j$ is an inclusive descendant of some $u_k \in N_+(u_i)$.
\end{definition}

\begin{definition}[Rooted subembedding]
  \label{def:subembedding}
  Let $(u_i, v)$ be a candidate vertex, $I$ be the set of all the inclusive descendants of $u_i$, and $M$ be a set of assignments of $u_j \in I$.
  $M$ is a \emph{subembedding} rooted at $(u_i, v)$ if and only if
  (i) $M$ is an embedding of $Q[I]$,
  (ii) $M$ includes assignment $(u_i, v)$, and 
  \revision{R1}{D3}{}{
  (iii) $M(u_j) \in C(u_j)$ holds for arbitrary $u_j$.
  }
\end{definition}

\revisiontext{
Condition (iii) makes sense because $Q[I]$ may be able to be mapped to data vertices that are not candidate vertices for $Q$.
}

\begin{definition}[Reservation]
  \label{def:reservation}
  Let $(u_i, v)$ be a candidate vertex and $S \subseteq V_G$.
  $S$ is a \emph{reservation} of $(u_i, v)$ if and only if an arbitrary subembedding rooted at $(u_i, v)$ contains an assignment to a data vertex in $S$.
\end{definition}

In this definition, any superset of a reservation is also a reservation.
Additionally, \revision{R1}{D4}{}{an arbitrary set of data vertices}, including an empty set, is a reservation of candidate vertex $(u_i, v)$ if there does not exist any subembedding rooted at $(u_i, v)$.
For this reason, at least one reservation can be defined for every candidate vertex.

\begin{example}
  \label{ex:reservation}
  The inclusive descendants of $u_1$ are $u_1$, $u_2$, $u_3$, and $u_4$.
  As shown in the GCS (\cref{fig:overall-example-gcs}), the subembeddings rooted at $(u_1, v_3)$ are
    $\{ (u_1, v_3),\allowbreak (u_2, v_5),\allowbreak (u_3, v_9),\allowbreak (u_4, v_0) \}$,
    $\{ (u_1, v_3),\allowbreak (u_2, v_7),\allowbreak (u_3, v_{10}),\allowbreak (u_4, v_0) \}$,
    $\{ (u_1, v_3),\allowbreak (u_2, v_8),\allowbreak (u_3, v_{11}),\allowbreak (u_4, v_1) \}$, and
    $\{ (u_1, v_3),\allowbreak (u_2, v_8),\allowbreak (u_3, v_{12}),\allowbreak (u_4, v_1) \}$.
  All of them contain an assignment to $v_0$ or $v_1$, and thus $\{v_0, v_1\}$ is one of the reservations of $(u_1, v_3)$.
\end{example}

GuP selects one of the possible reservations of $(u_i, v)$ as reservation guard $R(u_i, v)$.
The definition of a reservation implies that, if a partial embedding contains assignments to $R(u_i, v)$, we cannot assign inclusive descendants of $u_i$ satisfying the injectivity constraint.
We formally state it as follows.
Let $M[\lcolon i]$ be $\{(u_j, v) \in M \mid j < i\}$.

\begin{definition}[Matching with a reservation guard]
We say that partial embedding $M$ \emph{matches} $R(u_i, v)$ if and only if $R(u_i, v) \subseteq \mathrm{Im}(M[\lcolon i])$ holds.
\end{definition}

\begin{lemma}
  \label{lem:reservation}
  Let $M$ be a partial embedding.
  Suppose that $R(u_i, v)$ is a reservation of $(u_i, v)$ and $M$ matches $R(u_i, v)$.
  Then, $M[\lcolon i] \cup \{(u_i, v)\}$ is a deadend.
\end{lemma}
\begin{proof}
  We make a proof by contradiction.
  Suppose that there exists full embedding $\hat{M}$ such that $R(u_i, v) \subseteq \mathrm{Im}(\hat{M}[\lcolon i])$ and $(u_i, v) \in \hat{M}$.
  Letting $I$ be the set of all the inclusive descendants of $u_i$, $\hat{M}[I]$ is a $(u_i, v)$-rooted subembedding.
  Hence, by \cref{def:reservation}, there exists $u_j \in I$ such that $\hat{M}(u_j) \in R(u_i, v)$.
  In addition, $j \ge i$ holds because $u_j$ is an inclusive descendant of $u_i$.
  At the same time, from $R(u_i, v) \subseteq \mathrm{Im}(\hat{M}[\lcolon i])$, there exists $u_k$ such that $k < i$ and $\hat{M}(u_k) = \hat{M}(u_j) \in R(u_i, v)$.
  Since $k \neq j$ holds by $k < i \le j$, $\hat{M}$ violates the injectivity constraint, which is a contradiction.
\end{proof}

Letting $i = |M|$ and $v \in C(u_i)$, $M \oplus v$ is a deadend if $M$ matches $R(u_i, v)$.
We can filter out $v$ from $C(u_i)$ by using this property.

\subsubsection{Reservation Guard Generation}

For effective pruning, we prefer that $\mathit{R}(u_i, v)$ is a reservation expected to be matched by many partial embeddings.
To generate such a reservation guard, GuP employs two policies.
First, GuP avoids generating a reservation guard that cannot be matched by any possible partial embedding.
We utilize the following lemma to check it.
Let $V_Q[\lcolon i] = \{ u_j \mid j < i \}$, $C^{-1}(v) = \{ u_i \mid v \in C(u_i) \}$, $C^{-1}(S) = \bigcup_{v \in S} C^{-1}(v)$, and $C^{-1}(v)[\lcolon i] = \{ u_j \in C^{-1}(v) \mid j < i \}$.

\begin{lemma}
  \label{lem:matchable}
  Let $(u_i, v)$ be a candidate vertex and $S$ be its reservation.
  If $S$ holds (i) $\exists v' \in S, C^{-1}(v')[\lcolon i] = \emptyset$ or (ii) $\exists S' \subseteq S, |S'| > |C^{-1}(S')[\lcolon i]|$, no partial embedding matches $S$.
\end{lemma}
\begin{proof}
  We make a proof by cases.
  Regarding the case that \textbf{$S$ holds condition (i)},
    suppose that $v' \in S$ holds $C^{-1}(v')[\lcolon i] = \emptyset$.
    Then, we have $v' \not\in C(V_Q[\lcolon i])$.
    Since arbitrary partial embedding $M$ holds $\mathrm{Im}(M[\lcolon i]) \subseteq C(V_Q[\lcolon i])$, we have $v' \not\in \mathrm{Im}(M[\lcolon i])$.
    It follows that $S \not\subseteq \mathrm{Im}(M[\lcolon i])$.
  Regarding the case that \textbf{$S$ holds condition (ii)},
    we argue by contradiction.
    Suppose there exists partial embedding $M$ such that $S \subseteq \mathrm{Im}(M[\lcolon i])$.
    Since $M$ is an injective mapping, we can have its inverse function $M^{-1}$.
    We have $M^{-1}(v') \in V_Q[\lcolon i]$ and $M^{-1}(v') \in C^{-1}(v')$ for all $v' \in \mathrm{Im}(M[\lcolon i])$, and thus $M^{-1}(v') \in C^{-1}(v')[\lcolon i]$ holds.
    Since $S \subseteq \mathrm{Im}(M[\lcolon i])$ holds by hypothesis, for all $S' \subseteq S$, we have $M^{-1}(S') \subseteq C^{-1}(S')[\lcolon i]$, and thus $|S'| \le |C^{-1}(S')[\lcolon i]|$ holds because $|M^{-1}(S')| = |S'|$.
    This contradicts supposition $\exists S' \subseteq S, |S'| > |C^{-1}(S')[\lcolon i]|$.
  Therefore, we have proved the statement.
\end{proof}

GuP considers a reservation to be \emph{matchable} if it satisfies neither condition (i) nor (ii) in \cref{lem:matchable}.

\begin{example}
  As shown in \cref{ex:reservation}, $\{v_0, v_1\}$ is a reservation of $(u_1, v_3)$.
  Let us check if this is matchable by using conditions (i) and (ii) in \cref{lem:matchable}.
  Both $C^{-1}(v_0)$ and $C^{-1}(v_1)$ are $\{u_0, u_4\}$, and thus $C^{-1}(v_0)[\lcolon 1] = C^{-1}(v_1)[\lcolon 1] = \{u_0\} \neq \emptyset$.
  Hence, condition (i) is not held.
  However, we have $|C^{-1}(\{v_0, v_1\})[\lcolon i]| = |\{u_0\}| = 1$, which is less than $|\{v_0, v_1\}|$ ($= 2$).
  Therefore, condition (ii) is held, and $\{v_0, v_1\}$ is not matchable as a reservation guard of $(u_1, v_3)$.
\end{example}

The second policy is to minimize the size of a reservation guard (i.e., $|R(u_i, v)|$) because a smaller reservation guard tends to be matched by more partial embeddings.
  By definition, an arbitrary candidate vertex $(u_i, v)$ has a \emph{trivial reservation} $\{v\}$ because any subembeddings rooted at $(u_i, v)$ contain $v$.
  Although this is clearly the smallest choice for $R(u_i, v)$, this does not reduce futile recursions because it just performs an ordinary injectivity check that ensures $v$ is not used in a partial embedding.
To obtain a non-trivial small reservation, GuP generates candidates of reservation guards by propagating the injectivity constraint in a bottom-up manner and selects the smallest one as a reservation guard.

\begin{definition}[Reservation guard candidate]
  \label{def:reservation-guard-candidate}
  Let $(u_i, v)$ be a candidate vertex, $u_j \in N_+(u_i)$, and $R(u_i, v)$ be a reservation of $(u_i, v)$.
  The \emph{reservation guard candidate} of $(u_i, v)$ regarding $u_j$ is the smallest set of data vertices
  \revision{R1}{D5}{}{that is (i) matchable as a reservation guard of $(u_i, v)$ and (ii) a superset of $\{v'\}$ or $R(u_j, v') \setminus \{v\}$ for all $v' \in N(v) \cap C(u_i)$.}
\end{definition}

\begin{lemma}
  \label{lem:bottomup-reservation-necessary-condition}
  Let $(u_i, v)$ be a candidate vertex and $u_j \in N_+(u_i)$.
  Suppose that $S$ is a reservation guard candidate of $(u_i, v)$ regarding $u_j$.
  Then, $S$ is a reservation of $(u_i, v)$.
\end{lemma}
\begin{proof}
  We argue by contradiction.
  Suppose that $S$ is not a reservation of $(u_i, v)$.
  It follows that there exists $(u_i, v)$-rooted subembedding $M$ that does not have an assignment to a vertex in $S$.
  Since $M$ satisfies the adjacency constraint, $M(u_j) \subseteq N(v) \cap C(u_j)$ holds.
  By hypothesis, we have $M(u_j) \in S$ or $R(u_j, M(u_j)) \setminus \{v\} \subseteq S$.
  Thus, we proceed by cases.
  Regarding \textbf{the case that $M(u_j) \in S$ holds},
    $M$ has an assignment to a vertex in $S$, which is a contradiction.
  Regarding \textbf{the case that $R(u_j, M(u_j)) \setminus \{v\} \subseteq S$ holds},
    let $I'$ be the inclusive descendant set of $u_j$.
    $M[I']$ has an assignment to a vertex in $R(u_j, M(u_j))$ because $M[I']$ is a subembedding rooted at $(u_j, M(u_j))$.
    In addition, since $M$ maps $u_i$ to $v$ and satisfies the injectivity constraint, $M[I']$ does not have an assignment to $v$.
    Hence, $M[I']$ has an assignment to a vertex in $R(u_j, M(u_j)) \setminus \{v\} \subseteq S$, which is a contradiction.
  Therefore, $S$ is a reservation of $(u_i, v)$.
\end{proof}

We can obtain reservation guard candidates via solving the vertex cover problem.

\begin{lemma}
  \label{lem:set-cover}
  Let $(u_i, v)$ be a candidate vertex, $u_j \in N_+(u_i)$, and
  \begin{equation}
    \label{eq:vertexcover-universe}
    E_R = \{ (v', w) \mid v' \in N(v) \cap C(u_j), w \in R(u_j, v') \setminus \{v\} \}.
  \end{equation}
  Consider graph $G_R = (V_R, E_R)$ where $V_R$ is a set of the data vertices that appear in $E_R$.
  In addition, suppose that $S \subseteq V_G$ holds that (i) $S$ is matchable as a reservation guard of $(u_i, v)$ and (ii) $S$ is the minimum vertex cover of $G_R$.
  Then, $S$ is the reservation guard candidate of $(u_i, v)$ regarding $u_j$.
\end{lemma}
\begin{proof}
  Let $v' \in N(v) \cap C(u_j)$ and $w \in R(u_j, v') \setminus \{v\}$.
  Since $S$ is a vertex cover, $S$ contains either or both of $v'$ and $w$.
  If $v' \not\in S$, by the definition of $E_R$, we have $w \in S$ for all $w \in R(u_j, v') \setminus \{v\}$.
  Hence, $S$ is a superset of $\{v'\}$ or $R(u_j, v') \setminus \{v\}$ for all $v' \in N(v) \cap C(u_j)$.
  Since $S$ is the smallest set that is matchable and covers $G_R$, it is the reservation guard candidate of $(u_i, v)$ regarding $u_j$.
\end{proof}

Since the reservation guard candidate must be matchable, it is not necessarily defined for all the forward neighbors.
If none of them has a reservation guard candidate, GuP resorts to a trivial reservation.
Otherwise, the smallest one is chosen for the reservation guard.
Specifically, a reservation guard is defined as follows.

\begin{definition}[Reservation guard]
  \label{def:reservation-guard}
  Let $(u_i, v)$ be a candidate vertex.
  The reservation guard on $(u_i, v)$, denoted by $R(u_i, v)$, is defined as follows.
  \begin{enumerate}
    \item If $N_+(u_i) = \emptyset$ or the reservation guard candidate of $(u_i, v)$ regarding $u_j$ is undefined for all $u_j \in N_+(u_i)$, $R(u_i, v) = \{v\}$.
    \item Otherwise, $R(u_i, v)$ is the smallest reservation guard candidate of $(u_i, v)$ regarding $u_j$ where $u_j \in N_+(u_i)$.
  \end{enumerate}
\end{definition}

\begin{example}
  \cref{fig:overall-example-gcs} shows reservation guard $R$ of each candidate vertex.
  For example, $R(u_2, v_5)$ is obtained from the reservation guards of forward-adjacent candidate vertices, $R(v_3, v_9)$, $R(u_4, v_0)$, and $R(u_4, v_{13})$, as follows.
  $R(u_4, v_0)$ and $R(u_4, v_{13})$ is given by the trivial reservations, $\{v_0\}$ and $\{v_{13}\}$, respectively.
  Since $u_3$ has only forward neighbor $u_4$, $R(u_3, v_9)$ is given by the reservation guard candidate regarding $u_4$.
  It is the smallest matchable superset of $\{v_0\}$ or $R(u_4, v_0) \setminus \{v_9\}$ ($= \{v_0\}$), and thus $R(u_3, v_9) = \{v_0\}$.
  Next, we consider the reservation guard candidates of $(u_2, v_5)$ regarding $u_3$ and $u_4$.
  That regarding $u_3$ is $R(u_3, v_9) \setminus \{v_5\}$ ($= \{v_0\}$) because $\{v_9\}$ is not matchable.
  Regarding $u_4$, there does not exist a matchable superset of both $\{v_0\}$ and $\{v_{13}\}$ because of condition (i) in \cref{lem:matchable} ($C^{-1}(v_{13})[\lcolon 2] = \emptyset$).
  Hence, the reservation guard of $(u_2, v_5)$ is given by the candidate regarding $u_3$, which is $\{v_0\}$.
\end{example}

\revisiontext{
\cref{def:reservation-guard} provides a reservation guard effective for pruning, but it still possibly produces a large reservation guard in theory.
  Such reservation guards are not only rarely matched by partial embeddings but also increase the cost of the reservation guard generation and matching test.
  Thus, we introduce user-defined parameter $r$ to specify the upper bound of the size of reservation guards.
  Fortunately, our empirical study in \cref{sec:parameter-r} showed that a reservation guard preserves its pruning power even if the size is limited to 3 (i.e., $r = 3$).
}

\begin{algorithm}[!t]
  \small
  \caption{Reservation guard generation}
  \begin{algorithmic}[1]
    \Require Candidate-vertex set $C$ and reservation size limit $r$
    \Ensure Reservation guard $R(u_i, v)$ for all $u_i \in V_Q$ and $v \in C(u_i)$
    \For{$u_i \in V_Q$ in reverse order and $v \in C(u_i)$} \label{ln:resv-gen-cand-loop}
      \State $R(u_i, v) \gets \{v\}$
      \For{$u_j \in N_+(u_i)$} \label{ln:resv-gen-nbr-loop}
        \State \revision{R1}{D6}{}{Construct graph $G_R$ from edge set $E_R$ (\cref{eq:vertexcover-universe})} \label{ln:resv-gen-u}
        \State Find vertex cover $S$ over $G_R$ s.t.~$|S| \le r$ and $S$ is matchable
        \If{such $S$ is found, and $R(u_i, v) = \{v\}$ or $|R(u_i, v)| < |S|$}
          \State $R(u_i, v) \gets S$
        \EndIf
      \EndFor
    \EndFor
  \end{algorithmic}
  \label{alg:reservation-guard-generation}
\end{algorithm}

\cref{alg:reservation-guard-generation} shows the pseudocode of the reservation guard generation.
This algorithm computes $R(u_i, v)$ in the descending order of $u_i$ so that the reservation guards of all the forward-adjacent candidate vertices are computed before $R(u_i, v)$.
The loop at line \ref{ln:resv-gen-nbr-loop} computes $S$, the reservation guard candidate of $(u_i, v)$ regarding $u_j \in N_+(u_i)$.
Since the vertex cover problem is NP-hard, GuP employs a 2-approximate algorithm described in \cite{Cormen2009}, which iteratively chooses an edge in $E_R$ and adds endpoints to $S$ until all the vertices are covered.
To find matchable $S$ whose size does not exceed $r$, our algorithm chooses an edge that keeps $S$ matchable by checking the conditions in \cref{lem:matchable} and stops the iteration if $|S|$ exceeds $r$.
If such $S$ is found for at least one of $u_j \in N_+(u_i)$, the smallest one is chosen for $R(u_j, v)$.
Otherwise, $R(u_j, v)$ is set to trivial reservation $\{v\}$.
We analyze the complexity of this algorithm in \cref{sec:complexity-analysis}.

\subsection{Nogood guard}
\label{sec:nogood-guard}

This section first defines a nogood guard and then details the nogood discovery rule for a nogood guard on vertices and edges, respectively.

\subsubsection{Definitions}

A nogood guards is a set of assignments that compose a nogood with a candidate vertex or the endpoints of candidate edges where it is attached.

\begin{definition}[Nogood]
\label{def:nogood}
Set of assignment $D \subseteq V_Q \times V_G$ is a \emph{nogood} if and only if there does not exist any full embedding $M$ such that $D \subseteq M$.
\end{definition}

\revision{R1}{O1}{example-nogood}{
For example, $\{ (u_0, v_0), (u_4, v_0) \}$ is a nogood because any partial embedding including these assignments violates the injectivity constraint.
}

\begin{definition}[Nogood guard on vertices]
\label{def:vertex-nogood}
\revision{R1}{O1}{def-vertex-nogood}{
Let $(u_i, v)$ be a candidate vertex.
A nogood guard on $(u_i, v)$, denoted by $\mathit{NV}(u_i, v)$, is a subset of $V_Q[\lcolon i] \times V_G$ such that $\mathit{NV}(u_i, v) \cup \{(u_i, v)\}$ is a nogood.
}
\end{definition}

\begin{definition}[Nogood guard on edges]
\label{def:edge-nogood}
\revision{R1}{O1}{def-edge-nogood}{
Let $((u_i, v), (u_j, v'))$ be a candidate edge.
A nogood guard on $((u_i, v), (u_j, v'))$, denoted by $\mathit{NE}((u_i, v), (u_j, v'))$, is a subset of $V_Q[\lcolon i] \times V_G$ such that $\mathit{NE}((u_i, v),\allowbreak (u_j, v')) \cup \{(u_i, v), (u_j, v')\}$ is a nogood.
}
\end{definition}

\begin{definition}[Matching with a nogood guard]
We say that partial embedding $M$ \emph{matches} $\mathit{NV}(u_i, v)$ or $\mathit{NE}((u_i, v), (u_j, v'))$ if and only if $M$ is a superset of them.
\end{definition}

Since a partial embedding including a nogood never yields any full embeddings, nogoods are useful for pruning.
Suppose that $M$ is a partial embedding of length $i$.
We can filter out candidate vertex $(u_i, v)$ if $M$ matches $\mathit{NV}(u_i, v)$ because $M \oplus v$ is a superset of $\mathit{NV}(u_i, v) \cup \{(u_i, v)\}$.
Similarly, we can filter out candidate edge $((u_i, v), (u_j, v'))$ if $M$ matches $\mathit{NE}((u_i, v), (u_j, v'))$.
\revisionx{R1}{D7}{}{
Filtering out the edge can be rephrased as prohibiting mapping query edge $(u_i, u_j)$ to data edge $(v, v')$ by two assignments $(u_i, v)$ and $(u_j, v')$.
}
Such mapping makes a deadend because $M \cup \{(u_i, v), (u_j, v')\}$ is a superset of $\mathit{NE}((u_i, v), (u_j, v')) \cup \{(u_i, v), (u_j, v')\}$, which is a nogood.

To generate a nogood guard, we need to discover a nogood from deadend partial embeddings encountered during the backtracking. 
The following sections describe how to do it.

\subsubsection{Nogood Guards on Vertices}
\label{sec:vertex-nogood-guard}

If a partial embedding is a deadend, the set of all its assignments is a nogood by the definition.
However, such nogoods are ineffective for pruning because the same partial embedding does not appear again during the search, and so no partial embedding matches it.
For higher effectiveness, a smaller nogood is preferred because it is expected to be a subset of more partial embeddings.
Thus, we need to carefully drop assignments irrelevant to the conflicts in a deadend.

A deadend violates some of the constraints of isomorphism.
To capture the adjacency constraint, we introduce a set of local candidate vertices, which satisfy the adjacency constraint between the assignments in a partial embedding.

\begin{definition}[Local candidate-vertex set]
  \label{def:local-candidate-vertex-set}
  Let $M$ be a partial embedding.
  The \emph{local candidate-vertex set} of $u_i$ under $M$, denoted by $C(u_i; M)$, is a set of $v \in C(u_i)$ that holds, for all $u_j \in N_-(u_i)$, (i) $v \in N(M(u_j))$ and (ii) $M$ does not match $\mathit{NE}((u_j, M(u_j)),\allowbreak (u_i, v))$.
\end{definition}

This definition uses a nogood guard on edges because we filter out candidate edges if $M$ matches their guards.
To satisfy the adjacency constraint, a local candidate vertex is used for extending $M$.
In other words, local candidate-vertex sets confine the search space.
Hence, the assignments involved in the computation of a local candidate-vertex set are relevant to generating a deadend.
Conversely, the assignments that have no influence on a local candidate-vertex set are irrelevant to a deadend, and so we can drop them to reduce the size of nogood.
Based on this idea, we define a bounding set, which is a set of query vertices whose assignments determine the local candidate-vertex sets.

\begin{definition}[Bounding set]
  \label{def:bounding-set}
  Let $M$ be a partial embedding and $u_i \in V_Q$.
  The \emph{bounding set of $u_i$ under $M$}, denoted by $B(u_i; M)$, is the union of $B_\mathrm{adj}$ and $B_\mathrm{guard}$.
  $B_\mathrm{adj}$ is the set of $u_j \in N_-(u_i)$ such that $C(u_i; M[\lcolon j]) \neq C(u_i; M[\lcolon j + 1])$.
  $B_\mathrm{guard}$ is the union of $\mathrm{dom}(\mathit{NE}((u_j, M(u_j)), (u_i, v)))$ for all combinations of $u_j \in N_-(u_i)$ and $v \in C(u_i; M[\lcolon j])$ such that $M$ matches $\mathit{NE}((u_j, M(u_j)),\allowbreak (u_i, v))$.
\end{definition}

$B_\mathrm{adj}$ is a set of query vertices whose assignment reduces the size of the local candidate-vertex set by its adjacency relation regardless of guards.
It is examined by condition $C(u_j; M[\lcolon i]) \neq C(u_j; M[\lcolon i + 1])$, which says that $u_i$ is included in the bounding set if adding assignment $(u_i, M(u_i))$ changes the local candidate-vertex set of $u_j$.
On the other hand, $B_\mathrm{guard}$ is a set of query vertices whose assignment commits in the match between $M$ and nogood guards.
They are also involved in the decision of the local candidate-vertex set because they are necessary to let $M$ match the nogood guards and filter out some edges.

\begin{example}
  \revision{R1}{O1}{example-bounding-set}{
  \label{ex:bounding-set}
  Let $M = \{ (u_0, v_0), (u_1, v_3) \}$ and assume that $M$ do not match any nogood guard on edges.
  Consider the bounding set of $u_2$ under $M$.
  We have $N_-(u_2) = \{u_0, u_1\}$.
  $B_\mathrm{adj}$ contains $u_0$ because $C(u_2; M[\lcolon 0]) = \{ v_5, v_6, v_7, v_8 \}$, which differs from $C(u_2; M[\lcolon 1]) = C(u_2; \{ (u_0, v_0) \}) = \{ v_5, v_6, v_7 \}$.
  On the other hand, $B_\mathrm{adj}$ does not contain $u_1$ because $N(M(u_1))$ is a superset of $C(u_2; M[\lcolon 1])$ and thus $C(u_2; M[\lcolon 2]) = C(u_2; M[\lcolon 1])$.
  In addition, $B_\mathrm{guard} = \emptyset$ by assumption.
  Therefore, we have $B(u_2; M) = \{ u_0 \}$.
  }
\end{example}

\begin{lemma}
  \label{lem:bounding-set}
  Let $M$ and $M'$ be a partial embedding and $u_i \in V_Q$.
  If $M[B(u_i; M)] \subseteq M'$, we have $C(u_i; M) \supseteq C(u_i; M')$.
\end{lemma}
\begin{proof}
  We make a proof by contradiction.
  Suppose that there exists $v \in C(u_i)$ such that $v \in C(u_i; M')$ and $v \not\in C(u_i; M)$.
  From $v \not\in C(u_i; M)$, there exists $u_j \in N_-(u_i) \cap \mathrm{dom}(M)$ that holds either or both of (i) $v \not\in N(M(u_j))$ and that (ii) $\mathit{NE}((u_j, M(u_j)), (u_i, v))$ is matched by $M$.
  Regarding \textbf{case (i)},
    we have $u_j \in B(u_i; M)$ by \cref{def:bounding-set} because $C(u_i; M[\lcolon j]) \neq C(u_i; M[\lcolon j + 1])$.
    Since $M[B(u_i; M)] \subseteq M'$ by hypothesis, $M(u_j) = M'(u_j)$ holds.
    However, from $v \in C(u_i; M')$, we have $v \in N(M'(u_j)) = N(M(u_j))$, which is contradiction.
  Regarding \textbf{case (ii)},
    by \cref{def:bounding-set}, we have $\mathrm{dom}(\mathit{NE}((u_j, M(u_j)),\allowbreak (u_i, v))) \subseteq B(u_i; M)$.
    By hypothesis, $M[B(u_i; M)] \allowbreak\subseteq M'$ holds, and thus we have $\mathit{NE}((u_j, M(u_j)),\allowbreak (u_i, v)) \subseteq M'[\lcolon j]$.
    However, from $v \in C(u_i; M')$, $\mathit{NE}((u_j, M(u_j)),\allowbreak(u_i, v))$ is not matched by $M'$, which is contradiction.
  Therefore, we have $C(u_i; M) \supseteq C(u_i; M')$.
\end{proof}

In backtracking, each local candidate vertex is further checked whether it causes a conflict in a partial embedding.

\begin{definition}[Conflict]
  \label{def:conflict}
  Let $M$ be a partial embedding, $k$ be the length of $M$, and $v \in C(u_k; M)$.
  Extension $M \oplus v$ has a \emph{conflict} if any of the following conditions hold.
  \begin{enumerate}
    \item \emph{Injectivity conflict}: $M$ has an assignment to $v$.
    \item \emph{Reservation-guard conflict}: $M$ matches $R(u_k, v)$.
    \item \emph{Nogood-guard conflict}: $M$ matches $\mathit{NV}(u_k, v)$.
    \item \emph{No-candidate conflict}: There exists $u_i$ ($i > k$) that has no local candidate vertices under $M \oplus v$ (i.e., $C(u_i; M \oplus v) = \emptyset$).
  \end{enumerate}
\end{definition}

We can discover a nogood from extensions if they have conflicts.
To indicate assignments that constitute a nogood, we introduce \emph{mask} $K \subseteq V_Q$ that gives a nogood of partial embedding or extension $M$ as $M[K]$, where $M[K] = \{ (u_i, v) \in M \mid u_i \in K \}$.
The following definition gives the mask for extensions that have conflicts.

\begin{definition}[Conflict mask]
  \label{def:conflict-mask}
  Let $M$ be a partial embedding, $k$ be the length of $M$, and $v \in C(u_k, M)$.
  The \emph{conflict mask} of extension $M \oplus v$ is $\emptyset$ if the extension has no conflict; otherwise, it is defined for each conflict case as follows.
  \begin{enumerate}
    \item \emph{Injectivity conflict.}
      If $u_i \in \mathrm{dom}(M)$ holds $v = M(u_i)$, the conflict mask is $\{ u_i, u_k \}$.
    \item \emph{Reservation-guard conflict.}
      If $M$ matches $R(u_k, v)$, the conflict mask is $\{ u_j \mid \exists v' \in R(u_i, v), (u_j, v') \in M \} \cup \{ u_k \}$.
    \item \emph{Nogood-guard conflict.}
      Suppose that $M$ matches $\mathit{NV}(u_k, v)$.
      Then, the conflict mask is $\mathrm{dom}(\mathit{NV}(u_k, v)) \cup \{ u_k \}$.
    \item \emph{No-candidate conflict.}
      Suppose that $u_i$ has no local candidate vertices under $M \oplus v$.
      Then, the conflict mask is $B(u_i; M \oplus v)$.
  \end{enumerate}
\end{definition}

\begin{example}
  \label{ex:conflict-mask}
  \revision{R1}{O1}{example-conflict-mask}{
  Let $M = \{(u_0, v_0),\allowbreak(u_1, v_2),\allowbreak(u_2, v_6) \}$ and assume that $M$ do not match any nogood guards on edges.
  $M \oplus v_{11}$ has the no-candidate conflict because $v_6$ and $v_{11}$ do not have any common neighbor in $C(u_4)$ and so $C(u_4; M \oplus v_{11}) = \emptyset$.
  Therefore, the conflict mask of $M \oplus v_{11}$ is $B(u_4; M \oplus v_{11}) = \{ u_2, u_3 \}$.
  }
\end{example}

\begin{lemma}
  \label{lem:nogood-of-conflicting-extensions}
  Let $M$ be a partial embedding, $k$ be the length of $M$, $v \in C(u_k, M)$, and $K$ be the conflict mask of $M \oplus v$.
  If $M \oplus v$ has conflicts, $(M \oplus v)[K]$ is a nogood.
\end{lemma}
\begin{proof}
  We make a proof for each conflict case.
  Regarding \textbf{the injectivity conflict},
    suppose that $K = \{u_i, u_k\}$. By \cref{def:conflict-mask}, we have $M[u_i] = v$, and thus $(M \oplus v)[K] = \{(u_i, v), (u_k, v)\}$.
    This is a nogood because of the violation of the injectivity constraint.
  Regarding \textbf{the reservation-guard conflict}, 
    suppose that there exists full embedding $\hat{M}$ such that $(M \oplus v)[K] \subseteq \hat{M}$.
    We have $\mathrm{Im}(\hat{M}[\lcolon k]) \supseteq \mathrm{Im}(\hat{M}[K]) = \mathrm{Im}(M[K]) \supseteq R(u_k, v)$.
    In addition, $(u_k, v) \in \hat{M}$ because $u_k \in K$.
    From \cref{lem:reservation}, $\hat{M}[\lcolon k] \cup \{(u_k, v)\}$ is a nogood.
    Thus, $\hat{M}$ includes a nogood, which is a contradiction.
    Therefore, $(M \oplus v)[K]$ is a nogood.
  Regarding \textbf{the nogood-guard conflict},
    we have $(M \oplus v)[K] = \mathit{NV}(u_k, v) \cup \{(u_k, v)\}$, which is a nogood by \cref{def:vertex-nogood}.
  Regarding \textbf{the no-candidate conflict},
    suppose that there exists full embedding $\hat{M}$ such that $(M \oplus v)[K] \subseteq \hat{M}$.
    From $B(u_i; M \oplus v) \subseteq K$, we have $(M \oplus v)[B(u_i; M \oplus v)] \subseteq \hat{M}$.
    Here \cref{lem:bounding-set} gives $C(u_i; \hat{M}) \subseteq C(u_i; M \oplus v) = \emptyset$, which is contradiction.
    Hence, $(M \oplus v)[K]$ is a nogood.
  We have shown that $(M \oplus v)[K]$ is a nogood for all the cases.
\end{proof}

The conflict mask defines a nogood discovered from extensions with conflicts.
However, even if an extension is free from a conflict, it can be found to be a deadend if it fails to yield a full embedding in subsequent recursions.
Hence, we generalize the definition to extensions with and without conflicts.

\begin{definition}[Deadend mask]
  \label{def:deadend-mask}
  Let $M$ be an extension and $k$ be the length of $M$.
  In addition, for any $v'$, suppose that $K_{v'}$ is the deadend mask of $M \oplus v'$.
  The \emph{deadend mask} of $M$ is given by $K$ defined as follows.
  The cases are listed in order of priority.
  
  \begin{enumerate}
    \item If $M$ is not a deadend, $K = \emptyset$.
    \item If $M$ has a conflict, $K$ is the conflict mask of $M$.
    \item If some $v' \in C(u_k; M)$ holds $u_k \not\in K_{v'}$, $K = K_{v'}$.
    \item Otherwise, $K = \bigcup_{v' \in C(u_k; M)} K_{v'} \cup B(u_k; M) \setminus \{ u_k \}$.
  \end{enumerate}
\end{definition}

\begin{example}
  \label{ex:deadend-mask}
  \revision{R1}{O1}{example-deadend-mask}{
  Let $M = \{ (u_0, v_0), (u_1, v_2), (u_2, v_6) \}$.
  $M$ does not have a conflict but is a deadend because $C(u_3; M) = \{ v_{11} \}$ and $M \oplus v_{11}$ has the no-candidate conflict.
  Let $K_{v_{11}}$ be the deadend mask of $M \oplus v_{11}$.
  $K_{v_{11}}$ is given by the conflict mask of $M \oplus v_{11}$, which is $\{ u_2, u_3 \}$ (\cref{ex:conflict-mask}).
  Hence, the deadend mask of $M$ is $K_{v_{11}} \cup B(u_3; M) \setminus \{u_3\} = \{ u_2, u_3 \} \cup \{u_2, u_3 \} \setminus \{u_3\} = \{u_2\}$.
  }
\end{example}

\begin{lemma}
  \label{lem:nogood-of-deadends}
  Let $M$ be an extension and $K$ be the deadend mask of $M$.
  If $M$ is a deadend, $M[K]$ is a nogood.
\end{lemma}
\begin{proof}
  We make a proof by cases for each of cases (1)--(4) in \cref{def:deadend-mask}.
  Regarding \textbf{case (1)},
    we ignore this case because $M$ is a deadend by hypothesis.
  Regarding \textbf{case (2)},
    by \cref{lem:nogood-of-conflicting-extensions}, $M[K]$ is a nogood.
  Regarding \textbf{case (3)},
    we prove this case by induction.
    The base case is case (2).
    Suppose that $v' \in C(u_k; M)$ gives $K_{v'}$ such that $u_k \not\in K_{v'}$, and $(M \oplus v')[K_{v'}]$ is a nogood.
    From $K = K_{v'}$ and $u_k \not\in K_{v'}$, we have $(M \oplus v')[K_{v'}] = M[K]$.
    Thus, $M[K]$ is a nogood.
  Regarding \textbf{case (4)},
    similar to case (3), we use induction using case (2) as the base case; suppose that, for all $v' \in C(u_k; M)$, $(M \oplus v')[K_{v'}]$ is a nogood.
    Moreover, we use a proof by contradiction; suppose that there exists full embedding $\hat{M}$ such that $M[K] \subseteq \hat{M}$.
    Then, we have $\hat{M}(u_k) \in C(u_k; \hat{M}) \subseteq C(u_k; M)$ because $B(u_k; M) \subseteq K$ and \cref{lem:bounding-set}.
    Since $K \cup \{u_k\}$ includes the deadend mask of $M \oplus v'$ for all $v' \in C(u_k; M)$, it also includes the deadend mask of $M \oplus \hat{M}(u_k)$.
    Thus, $(M \oplus \hat{M}(u_k))[K \cup \{u_k\}]$ is a nogood.
    Since $(M \oplus \hat{M}(u_k))[K \cup \{u_k\}] = \hat{M}[K \cup \{u_k\}]$, $\hat{M}$ includes a nogood, which is contradiction.
    Therefore, $M[K]$ is a nogood.
  We have proved that $M[K]$ is a nogood for all the cases.
\end{proof}

GuP discovers a nogood from deadends by using the deadend mask.
Specifically, when extension $M$ has a conflict or is determined to be a deadend after the exploration, GuP obtains nogood $M[K]$ where $K$ is the deadend mask of $M$.
Then, letting $(u_i, v)$ be the last assignment in $M[K]$, GuP records $M[K] \setminus \{(u_i, v)\}$ in $\mathit{NV}(u_i, v)$.
Such $\mathit{NV}(u_i, v)$ holds \cref{def:vertex-nogood}.
Note that $\mathit{NV}(u_i, v)$ is overwritten if it has an old value.
In this way, GuP generates nogood guards on vertices during the backtracking.

\begin{example}
  \label{ex:vertex-nogood-guard}
  \revision{R1}{O1}{example-nogood-guard}{
  Let $M = \{ (u_0, v_0), (u_1, v_2), (u_2, v_6) \}$.
  Since deadend mask $K$ of $M$ is $\{u_2\}$ (\cref{ex:deadend-mask}), GuP records $M[\{u_2\}] \setminus \{(u_2, v_6)\} = \emptyset$ in $\mathit{NV}(u_2, v_6)$.
  Note that $\emptyset$ is a subset of an arbitrary set, and hence $(u_2, v_6)$ is never used in the subsequent search.
  }
\end{example}

As shown in \cref{ex:vertex-nogood-guard}, GuP can filter out unnecessary candidate vertices for all partial embeddings, besides adaptive filtering depending on a partial embedding.
This is because GuP can capture a lack of cycles during backtracking.
To the best of our knowledge, the existing candidate filtering methods cannot capture cycles.

\subsubsection{Nogood Guards on Edges}

Nogood guards on edges are used for filtering out candidate edges and reduces the number of local candidate vertices as shown in \cref{def:local-candidate-vertex-set}.
  This is beneficial because we can detect the no-candidate conflict in earlier backtracking steps if all the local candidates are filtered out.
  In particular, as mentioned in \cref{sec:introduction}, the cycles in a query graph must be mapped to cycles in a data graph to satisfy the adjacency constraint, although cycles are difficult to find because of the sparseness of graphs.
  Such a search tends to involve many no-candidate conflicts, and thus by detecting them earlier we can improve the search performance.

As defined in \cref{def:edge-nogood}, $\mathit{NE}((u_i, v),\allowbreak (u_j, v'))$ is a set of assignments such that $\mathit{NE}((u_i, v),\allowbreak (u_j, v')) \subseteq V_Q[\lcolon i] \times V_G$ holds and $D$ is a nogood, where $D = \mathit{NE}((u_i, v),\allowbreak (u_j, v')) \cup \{ (u_i, v),\allowbreak (u_j, v') \}$.
This definition prohibits that $D$ contains an assignment of any $u_k$ such that $i < k < j$.
A nogood discovered with a deadend mask may violate it, and hence we need another rule to discover a nogood for a nogood guard on edges.

For conciseness of the discussion, we relax the format of the nogood as follows.
Assume that $M$ is a partial embedding whose length is $i + 1$ (i.e., $M$ includes an assignment of $u_i$).
Then, our goal is to find mask $K \subseteq V_Q$ such that $M[K] \cup \{(u_j, v')\}$ is a nogood.
  It allows us to discuss this problem for an arbitrary combination of partial embedding $M$ and candidate vertex $(u_j, v')$ regardless of the existence of candidate edge $((u_i, M(u_i)), (u_j, v'))$.
We formally define such mask $K$.

\begin{definition}[Fixed deadend mask]
  \label{def:fixed-deadend-mask}
  Let $M$ be an extension, $k$ be the length of $M$, and $(u_i, v)$ be a candidate vertex.
  In addition, for any $v'$, suppose that $K_{v'}$ is the $(u_i, v)$-fixed deadend mask of $M \oplus v'$.
  The \emph{$(u_i, v)$-fixed deadend mask} of $M$ is given by $K$ defined as follows.
  The cases are listed in order of priority.
  
  \begin{enumerate}
    \item If $i < k$ holds, $K = K' \setminus \{ u_i \}$ where $K'$ is the deadend mask of $M[\lcolon i] \oplus v$.
    \item If some full embedding includes $M \cup \{(u_i, v)\}$, $K = \emptyset.$
    \item If $M$ has a conflict, $K$ is the conflict mask of $M$.
    \item If $M$ holds $v \not\in N(M(u_j))$ for some $u_j \in N(u_i)$, $K = \{u_j\}$\footnote{If multiple $u_j$ holds the condition, $u_j$ of the smallest $j$ is chosen.}.
    \item If $M$ matches $\mathit{NE}((u_j, v'), (u_i, v))$ for some $(u_j, v') \in M$, $K = \mathrm{dom}(\mathit{NE}((u_j, v'), (u_i, v))) \cup \{u_j\}.$\footnotemark[1]
    \item If some $v' \in C(u_k; M)$ holds $u_k \not\in K_{v'}$, $K = K_{v'}$.
    \item Otherwise, $K = \bigcup_{v' \in C(u_k; M)} K_{v'} \cup B(u_k; M) \setminus \{ u_k \}$.
  \end{enumerate}
\end{definition}

The definition of the $(u_i, v)$-fixed deadend mask resembles that of the deadend mask.
The main differences are that (i) $K_{v'}$ is recursively given by $(u_i, v)$-fixed deadend mask, and (ii) it has conditions on $u_i$ and $v$ (cases (1), (4), and (5)).
Case (1) is the base case defined using the deadend mask.
Case (4) and (5) handle the case that $v$ is not a local candidate vertex of $u_i$.

\begin{example}
  \label{ex:fixed-deadend-mask}
  \revision{R1}{O1}{example-fixed-deadend-mask}{
  Let $M = \{ (u_0, v_0), (u_1, v_2), (u_2, v_7) \}$ and consider the $(u_4, v_0)$- and $(u_4, v_1)$-fixed deadend masks of $M$.
  Since $M \oplus v_{10} \oplus v_0$ has the injectivity conflict, its $(u_4, v_0)$-fixed deadend mask is $\{u_0, u_4\}$ by case (3) of \cref{def:fixed-deadend-mask}.
  Then, $(u_4, v_0)$-fixed deadend mask of $M \oplus v_{10}$ is $\{u_0, u_4\} \cup B(u_4; M \oplus v_{10}) \setminus \{u_4\} = \{u_0, u_2, u_3\}$ by case (7).
  It follows that $(u_4, v_0)$-fixed deadend mask of $M$ is $\{u_0, u_2, u_3\} \cup B(u_3; M) \setminus \{u_3\} = \{u_0, u_2\}$.
  On the other hand, the $(u_4, v_1)$-fixed deadend mask of $M \oplus v_{10}$ is $\{u_3\}$ by case (4) because $v_{10} \not\in v_{10}$.
  Hence, by case (7), $(u_4, v_1)$-fixed deadend mask of $M$ is $\{u_3\} \cup B(u_3; M) \setminus \{u_3\} = \{ u_2 \}$.
  } 
\end{example}

\begin{lemma}
  \label{lem:nogood-of-fixed-deadends}
  Let $M$ be an extension, $(u_i, v)$ be a candidate vertex, and $K$ be the $(u_i, v)$-fixed deadend mask of $M$.
  Suppose that $|M| \le i$ and $M \cup \{(u_i, v)\}$ is a nogood.
  Then, $M[K] \cup \{(u_i, v)\}$ is a nogood.
\end{lemma}
\begin{proof}
  We make a proof by cases for cases (1), (4), and (5) in \cref{def:fixed-deadend-mask} and omit the others because they can be proved similarly to the proof of \cref{lem:nogood-of-deadends}.
  Regarding \textbf{case (1)},
    let $K'$ be the deadend mask of $M[\lcolon i] \oplus v$.
    We have $(M[\lcolon i] \oplus v)[K'] = M[\lcolon i][K' \setminus \{u_i\}] \cup \{ (u_i, v) \} = M[K] \cup \{(u_i, v)\}$ because $M[\lcolon i] = M$ holds by hypothesis $|M| \le i$.
    Hence, $M[K] \cup \{(u_i, v)\}$ is a nogood.
  Regarding \textbf{case (4)},
    we have $M[K] \cup \{(u_i, v)\} = \{(u_j, M(u_j), (u_i, v)\}$.
    By hypothesis, $(u_j, u_i) \in E_Q$ and $(M(u_j), v) \not\in E_G$ hold, and thus $M[K] \cup \{(u_i, v)\}$ is a nogood because of the violation of the adjacency constraint.
  Regarding \textbf{case (5)},
    Let $M'$ be an arbitrary partial embedding such that $M[K] \cup \{(u_i, v)\} \subseteq M'$.
    Since $\mathrm{dom}(\mathit{NE}((u_j, M(u_j)),\allowbreak (u_i, v))) \subseteq K$ holds, we have $\mathit{NE}((u_j, M(u_j)),\allowbreak (u_i, v)) \subseteq M'$, and thus $M'[\lcolon j] \cup \{(u_j, M(u_j)),\allowbreak (u_i, v)\}$ is a nogood by \cref{def:edge-nogood}
    From $u_j \in K$, we have $M'(u_j) = M(u_j)$.
    Hence, $M'[\lcolon j] \cup \{(u_j, M(u_j)),\allowbreak (u_i, v)\} \subseteq M'$ holds, which means $M'$ is a deadend.
    It follows that $M[K] \cup \{(u_i, v)\}$ is a nogood.
  We have proved the statement.
\end{proof}

During backtracking, GuP updates a nogood guard on edges as follows.
Suppose that $M$ is a partial embedding of length $i$, $C(u_i; M)$ contains $v$, and there exists candidate edge $((u_i, v), (u_j, v'))$.
If extension $M \oplus v$ did not yield a full embedding containing $(u_j, v')$ in the subsequent recursions, GuP computes the $(u_j, v')$-fixed deadend mask $K$ of $M \oplus v$.
Then, GuP records $M[K]$ in $\mathit{NE}((u_i, v), (u_j, v'))$.
This holds \cref{def:edge-nogood} because $M[K] \cup \{(u_i, v), (u_j, v')\} = (M \oplus v)[K] \cup \{(u_j, v')\}$, which is a nogood by \cref{lem:nogood-of-fixed-deadends}.

\begin{example}
  \label{ex:edge-nogoo-guard}
  \revision{R1}{O1}{example-edge-nogood}{
  Let $M = \{ (u_0, v_0), (u_1, v_2) \}$.
  Since $(u_4, v_0)$-fixed deadend mask of $M \oplus v_7$ is $\{u_0, u_2\}$ (\cref{ex:fixed-deadend-mask}), GuP records $M[\{u_0, u_2\}]$ ($= \{(u_0, v_0)\}$) in $\mathit{NE}((u_2, v_7), (u_4, v_0))$.
  In addition, since $(u_4, v_1)$-fixed deadend mask of $M \oplus v_7$ is $\{u_2\}$, GuP records $M[\{u_2\}]$ ($= \emptyset$) in $\mathit{NE}((u_2, v_7), (u_4, v_1))$, which filters out $((u_2, v_7), (u_4, v_1))$ for all partial embeddings.
  }
\end{example}

Since the generation of a nogood guard incurs slight overheads, we optimized our implementation by utilizing nogood guards on edges only in the 2-core of a query graph.
As mentioned above, a nogood guard on edges is effective for pruning in the search of cyclic structures.
The subgraph outside of the 2-core consists of trees, and thus we did not generate the nogood guards on candidate edges that correspond to query edges outside of the 2-core.
This optimization allows GuP to profit from pruning opportunities offered by nogood guards on edges without sacrificing the efficiency.

\subsection{Backtracking with Guards}
\label{sec:backtracking}

\begin{algorithm}[!t]
  \small
  \caption{$\textproc{Backtrack}$ function of GuP}
  \begin{algorithmic}[1]
    \Require Partial embedding $M$ and sets of local candidate vertices $C_M$.
    \Ensure All the embeddings of $Q$ in $G$
    \IfThen{$|M| = |V_Q|$}{output $M$}
    \State $k \gets |M|$
    \For{$v \in C_M(u_k)$} \label{line:local-cand-loop}
      \IfThen{$M(u_i) = v$ for some $u_i$}{\textbf{continue}} \label{ln:outline-validation}
      \IfThen{$M$ matches $R(u_k, v)$ or $\mathit{NV}(u_k, v)$}{\textbf{continue}} \label{ln:outline-guard}
      \State Copy $C_M$ to $C'_M$
      \For{$u_i \in N_+(u_k)$}
        \State $C'_M(u_i) \gets \{ v' \in C_M(u_i) \cap N(v)$\label{ln:outline-refinement}
        \Statex \hspace{7em}$\mid M \text{ does not match } \mathit{NE}((u_k, v), (u_i, v'))\}$
      \EndFor
      \If{$C'_M(u_i) \neq \emptyset$ for all $u_i$} \label{ln:outline-nocandidate}
        \State $\textproc{Backtrack}(M \oplus v, C'_M)$ \label{ln:outline-recursion}
        \State Update $\mathit{NE}$ for each candidate edge incident to $(u_k, v)$
      \EndIf
      \If{$M \oplus v$ is found to be a deadend}
        \State Discover a nogood in $M \oplus v$ as $D$ and update $\mathit{NV}$
        \IfThen{$D \subseteq M$}{\Return} \Comment{Backjumping}\label{ln:backjumping}
      \EndIf
    \EndFor
  \end{algorithmic}
  \label{alg:backtrack}
\end{algorithm}

By utilizing a reservation guard and a nogood guard, GuP efficiently performs backtracking.
\cref{alg:backtrack} shows the backtracking algorithm of GuP.
Calling $\textproc{Backtrack}(\emptyset, C)$ starts the search, where $C$ is sets of the candidate vertices.
Function $\textproc{Backtrack}(M, C_M)$ recursively extends partial embedding $M$ until it obtains a full embedding.
  $C_M$ is sets of local candidate vertices under $M$ (i.e., $C_M(u_i) = C(u_i; M)$).
  If $M$ is not a full embedding, \textproc{Backtrack} tries extending $M$ with $v \in C_M(u_k)$.
  It first checks if $v$ is already assigned in $M$, or $M$ matches $R(u_k, v)$ or $\mathit{NV}(u_k, v)$.
  If so, $v$ is filtered out.
  Next, $C_M$ is refined to $C'_M$, the sets of local candidate vertices under $M \oplus v$, by following \cref{def:local-candidate-vertex-set}. 
    The bounding sets of each query vertex are also incrementally computed similar to $C_M$, while it is not shown in the pseudocode.
  If every query vertex retains at least one local candidate vertex after the refinement, \textproc{Backtrack} is recursively called.
  After the recursion, nogood guards are updated.
  In addition, discovered nogood $D$ is used for backjumping; specifically, $M \oplus v$ for any $v$ becomes a deadend if $M$ includes nogood $D$, and hence the iteration over $v \in C_M(u_k)$ is terminated.
Guards improve the search efficiency by filtering out unnecessary candidate vertices in $C_M(u_k)$ at line \ref{ln:outline-guard} and in $C'_M(u_i)$ for $u_i \in N_+(u_k)$ at line \ref{ln:outline-refinement}.

\begin{example}
  \label{ex:backtracking}
  \revision{R1}{O1}{example-backtracking}{
  The process of the backtracking can be considered as a depth-first search on a \emph{search tree} whose node corresponds to a recursive call with an extension.
  \cref{fig:overall-example-search-tree} shows the search tree of conventional backtracking.
  X-marks indicates a conflicting assignment.
  \cref{fig:overall-example-gcs} shows all the reservation guards and the nogood guards on vertices when the backtracking search reaches search node $m_6$, which corresponds to $M_6 = \{ (u_0, v_0), (u_1, v_3) \}$.
  Our backtracking algorithm now tries extending $M_6$ with each $v \in C(u_2; M_6)$ ($= \{ v_5, v_6, v_7 \}$), but all of them are filtered out by $R(u_2, v_5)$, $\mathit{NV}(u_2, v_6)$, and $\mathit{NV}(u_2, v_7)$, respectively.
  Hence, the function returns to node $m_1$, which corresponds to $M_1 = \{(u_0, v_0)\}$.
  Since $M_1 \oplus v_3$ ($= M_6$) was found to be a deadend, GuP computes deadend mask $K$ of $M_6$.
  From $K_{v_{5}} \cup K_{v_{6}} \cup K_{v_{7}} = \{u_0, u_2\}$ and $B(u_2; M_6) = \{u_0\}$ (\cref{ex:bounding-set}), $K = \{u_0, u_2\} \cup \{u_0\} \setminus \{u_2\} = \{u_0\}$ by case (4) of \cref{def:deadend-mask}.
  Thus, GuP discover nogood $D = M_6[K] = \{(u_0, v_0)\}$.
  Since $D \subseteq M_1$ holds, the function performs backjumping to the caller (line \ref{ln:backjumping}).
  This prunes search node $m_{12}$.
  As a whole, our approach prunes the shadowed nodes in \cref{fig:overall-example-search-tree}.
  }
\end{example}

\paragraph{Comparison with Failing Set-based Pruning}

\revision{R1}{D1}{method}{
Failing set-based pruning proposed in DAF \cite{Han2019} is one of the backjumping \cite{Rossi2006} methods and is popular in the database community \cite{Sun2020sigmod, Sun2021, Kim2022}.
\revisiontagx{R3}{O1}{failingset-method}
Although both DAF and GuP exploit nogoods for pruning, GuP is more effective for two reasons.
  First, GuP reuses a discovered nogood for pruning multiple times, whereas DAF discards a nogood after using it for backjumping.
  Note that GuP also performs backjumping (line \ref{ln:backjumping} in \cref{alg:backtrack}).
  Second, GuP discovers smaller nogoods, which offer higher pruning power.
    Like a deadend mask, DAF discovers a nogood using a failing set, which is defined as a set of query vertices and all their ancestors in terms of the matching order.
    Owing to the ancestors, a failing set tends to be large and so offers a large nogood.
    For example, a failing set of $M_6 = \{(u_0, v_0), (u_1, v_3)\}$ is $\{u_0, u_1\}$, and this fails to trigger a backjumping at search node $m_1$.
    On the other hand, GuP produces small deadend mask $\{u_0\}$ and can prune search node $m_{12}$ by the backjumping (\cref{ex:backtracking}).
    Thus, our nogood discovery rule enables more effective pruning.
}

\subsection{Optimizations}
\label{sec:optimizations}

In this section, we present additional techniques for improving the performance of subgraph matching.

\subsubsection{Search-node Encoding}
\label{sec:search-node-encoding}

The matching test of nogood guards takes nonnegligible computational costs.
Consider the matching test between partial embedding $M$ and $\mathit{NV}(u_i, v)$.
It takes $O(|\mathit{NV}(u_i, v)|)$ time to check $M(u_j) = v'$ for each $(u_j, v') \in \mathit{NV}(u_i, v)$.
This is the same for a nogood guard on edges.
The size of a nogood guard can be up to $|V_Q| - 1$, and GuP performs the matching test many times for filtering out candidate vertices and edges.
Thus, this overhead may spoil the performance benefit resulting from pruning.

To mitigate the overhead, we introduce a search-node encoding, which represent a nogood with a node in the search tree.
We assume that every node in the search tree has a unique ID number.
The search tree has a one-to-one correspondence between the nodes and the partial embeddings.
We refer to the node corresponding to partial embedding $M$ by the \emph{search node} of $M$.
Suppose that $m_i$ and $m_j$ are the search node of partial embeddings $M$ and $M'$.
Then, if $M'$ is an extension of $M$ (i.e., $M \subseteq M'$), $m_j$ is a descendant of $m_i$ in the search tree.
It can be checked in $O(1)$ time by maintaining the \emph{ancestor array} of $m_j$, denoted by $\mathit{anc}$.
Assuming that there exists an imaginary root node $m_0$, which corresponds to the empty partial embedding, $\mathit{anc}$ contains $m_0$ at $\mathit{anc}(0)$, a length-1 partial embedding at $\mathit{anc}(1)$, its child at $\mathit{anc}(2)$, and so on.
$\mathit{anc}(|M'|)$ is set to $m_j$.
Here, if $\mathit{anc}(|M|) = m_i$ holds, $m_j$ is a descendant of $m_i$.
This also means we can check if $M$ is a subset of $M'$ in $O(1)$ time.

\begin{example}
On the search tree shown in \cref{fig:overall-example-search-tree}, let us check if $M_3$ is a subset of $M_5$ where $M_3 = \{(u_0, v_0), (u_1, v_2), (u_2, v_6)\}$ and $M_5 = \{(u_0, v_0), (u_1, v_2), (u_2, v_7), (u_3, v_{10})\}$.
In the search tree, node $m_3$ and $m_5$ correspond to $M_3$ and $M_5$, respectively.
Ancestor array $\mathit{anc}$ of $m_5$ contains the IDs of $m_0$ (imaginary root node), $m_1$, $m_2$, $m_4$, $m_5$ in $\mathit{anc}(0)$ to $\mathit{anc}(4)$.
Since $\mathit{anc}(|M_3|) = \mathit{anc}(3) = m_4$ and thus $\mathit{anc}(|M_3|) \neq m_3$, we can find that $M_3$ is not a subset of $M_5$.
\end{example}

For applying this idea to matching tests with nogood guards, GuP encodes nogood guards into the ID of a search node.
Since a nogood guard is a subset of a partial embedding, it may not have a corresponding search node.
Thus, GuP ``rounds up'' a nogood guard to the minimum partial embedding including it.

\begin{definition}[Minimum superset embedding]
Let $M$ be a partial embedding and $D$ be a subset of $M$.
The \emph{minimum superset embedding} of $D$ in $M$ is $M[\lcolon i + 1]$ where $i$ is the query-vertex ID of the last assignment in $D$ (i.e., $D = \{\ldots, (u_i, v) \}$).
\end{definition}

In the definition above, letting $\mathit{anc}$ be the ancestor array of $M$, we can obtain the search node of minimum superset embedding $M[\lcolon i + 1]$ as $\mathit{anc}(i)$.
With rounding up to minimum superset embeddings, nogood guards can be encoded to the ID of the search node.

In summary, GuP stores nogood guards in a GCS as follows.
Suppose that nogood guard $D$ is extracted from partial embedding $M$, $L$ is the minimum superset embedding of $D$ in $M$, and $\mathit{anc}$ is the ancestor array of $M$.
Let $\mathit{id}$ be $\mathit{anc}(|L|)$, $\mathit{len}$ be $|L|$, and $K$ be $\mathrm{dom}(D)$.
Then, each nogood guard (both on vertices and edges) is stored as a triplet $(\mathit{id}, \mathit{len}, K)$.
Matching with partial embedding $M'$ is checked by $\mathit{anc}'(\mathit{len}) = \mathit{id}$ where $\mathit{anc}'$ is the ancestor array of $M'$.
$K$ is used to obtain $\mathrm{dom}(D)$ in the computation of bounding sets and the nogood discovery for the nogood-guard conflict case.
Thanks to the lightweight matching test with search-node encoding, GuP can efficiently filter out candidate vertices and edges.

\revisiontext{
\subsubsection{Parallelization}

\revisiontag{R3}{O3}{method}
Modern computers have multiple CPU cores and require parallel processing to utilize them.
We can easily parallelize backtracking of GuP, which tends to dominate query processing time, by searching different subtrees of the search tree in different threads.
Since the size of the search space is unknown in advance and usually very skewed, it is necessary to employ a work-stealing approach that dynamically splits the search tree and assigns it to an idle thread for load balancing.
Threads share the candidate vertices and edges and the reservation guards in a GCS but maintain thread-local nogood guards because those are modified during parallel backtracking.
Since the pruning efficiency may degrade because of not sharing information of nogoods between threads, we empirically evaluate it in \cref{sec:eval-parallel}.
}

\revisiontext{
\subsection{Complexity Analysis}
\label{sec:complexity-analysis}

\revisiontagx{R2}{O2}{complexity}
We first analyze the time complexity of each of three steps listed in \cref{sec:overview}.
\revisiontagx{R3}{O1}{complexity}, and then discuss the space complexity of the whole of GuP.
The following analyses assume that a bit vector of length $|V_Q|$ takes $O(1)$ space and $O(1)$ time for set operations, such as union and intersection, since a query graph is supposed to be small.

\textit{Time complexity of the GCS construction.}
GuP employs extended DAG-graph DP, which provides a candidate space through candidate filtering in $O(|E_Q| |E_G|)$ time \cite{Kim2022}.
Candidate filtering and GCS construction of GuP have the same complexity, $O(|E_Q| |E_G|)$, because we can obtain a GCS by attaching a null-valued guard to each candidate vertex and edge during extended DAG-graph DP.
GuP also adopts VC for optimizing the matching order, whose complexity is $O(|E_Q| |E_G|)$ \cite{Sun2020tkde}.
Therefore, the complexity of this step is $O(|E_Q| |E_G|)$.

\textit{Time complexity of the reservation guard generation.}
In the following, we show that \cref{alg:reservation-guard-generation} takes $O(|E_Q| |E_G|)$ time.
  Let $\bar{d}_Q$ and $\bar{d}_G$ be the average degrees of $Q$ and $G$, respectively.
  The loop over the candidate vertices (line \ref{ln:resv-gen-cand-loop}) iterates up to $|V_Q| |V_G|$ times, and the loop over forward neighbors of a query vertex (line \ref{ln:resv-gen-nbr-loop}) iterates $\bar{d}_Q$ times.
  The complexity for computing $E_R$ by \cref{eq:vertexcover-universe} (line \ref{ln:resv-gen-u}) is bounded by the size of $E_R$, which is $O(\sum_{v' \in N(v) \cap C(u_j)} |R(u_j, v')|) = O(|N(v)|\allowbreak \times r) = O(\bar{d}_G)$ since $r$ is a constant.
  After that, \cref{alg:reservation-guard-generation} solves the vertex-cover problem for graph $G_R = (V_R, E_R)$.
    $V_R$ consists of both endpoints of the edges, and thus $|V_R| \le 2 |E_R|$ holds.
    We employ the 2-approximation algorithm \cite{Cormen2009}, whose complexity is $O(|V_R| + |E_R|) = O(|E_R|) = O(|\bar{d}_G|)$.
  Therefore, the whole complexity of \cref{alg:reservation-guard-generation} is $O(|V_Q| |V_G| \bar{d}_Q \bar{d}_G) = O(|E_Q| |E_G|)$.
  
\textit{Time complexity of the backtracking search.}
Matching with a reservation guard takes $O(1)$ time because its size is bounded by $r$.
It also takes $O(1)$ time to perform matching with and generation of a nogood guard in search-node encoding as shown in \cref{sec:search-node-encoding}.
Hence, the complexity of backtracking is determined by the number of recursions.
While it is $O(\prod_{u_i} |C(u_i)|) = O(|V_G|^{|V_Q|})$ in the worst case \cite{Zhao2010}, the number significantly decreases in practice thanks to candidate filtering and guards.
However, theoretically analyzing their contribution is difficult because of their sensitivity to input graphs.
Thus, following previous studies \cite{Han2019, Sun2020sigmod, Sun2021, Kim2022}, we experimentally evaluate it in \cref{sec:evaluation}.

\textit{Space complexity.}
\revision{R2}{O4}{complexity}{
Since every part of GuP focuses only on candidate vertices and edges, a GCS dominates the space complexity of GuP.
A reservation guard consists of up to $r$ data vertices, and hence its size is regarded to be $O(1)$.
A nogood guard in search-node encoding also takes $O(1)$ space because it is a triplet of integers and a bit vector of query vertices, whose size is $O(1)$.
Therefore, the space complexity of a GCS is the same as a candidate space, which is $O(|E_Q| |E_G|)$ \cite{Han2019, Sun2020sigmod}.

}

\textit{Comparison with existing methods.}
\revision{R3}{O2}{complexity-comparison}{
If we leave out the exponential time complexity of backtracking, $O(|E_Q| |E_G|)$ is a common time and space complexity among recent methods \cite{Han2019, Bhattarai2019, Sun2021, Kim2022}.
GQL \cite{He2008} has an even higher time complexity due to semi-perfect matching \cite{Sun2020sigmod}.
However, the practical performance of subgraph matching largely depends on that of backtracking, and hence we experimentally show it in \cref{sec:evaluation}.
}
}

\section{Evaluation}
\label{sec:evaluation}

In this section, we compare the performance of GuP with existing methods and analyze GuP from various aspects.

\subsection{Experimental Setup}

\emph{Methods.}
We compared the performance of GuP with the following methods\footnote{We also tried to measure the performance of VEQ \cite{Kim2022}, but the binary obtained from \url{https://github.com/SNUCSE-CTA/VEQ} crashes during the process of over thousands of query graphs used in our experiment. Thus, we omitted its results for a fair comparison.}: DAF \cite{Han2019}, GQL-G \cite{Sun2020sigmod}, GQL-R \cite{Sun2020sigmod}, and RapidMatch (RM) \cite{Sun2021}.
All of them have been proposed in the last several years and employ failing set-based pruning.
GQL-G and GQL-R are combinations of candidate filtering of GraphQL \cite{He2008} and the matching orders of GraphQL and RI \cite{Bonnici2013}, respectively.
They performed the best in the evaluation by Sun et al.~\cite{Sun2020sigmod}.
Every implementation was obtained from the authors' GitHub repository\footnote{DAF: \url{https://github.com/SNUCSE-CTA/DAF}\\
GQL-G and GQL-R: \url{https://github.com/RapidsAtHKUST/SubgraphMatching}\\
RapidMatch: \url{https://github.com/RapidsAtHKUST/RapidMatch}}.
Our implementation of GuP\footnote{\url{https://github.com/araij/gup/}} employs candidate filtering with extended DAG-graph DP \cite{Kim2022} and the matching order produced by VC \cite{Sun2020tkde}.
We set $r$, the size limit of reservation guard, to $3$ unless otherwise specified.

\emph{Graphs.}
We used the following four data graphs: Yeast (3,112 vertices, 12,519 edges, 71 labels), Human (4,674 vertices, 86,282 edges, 44 labels), WordNet (76,853 vertices, 120,399 edges, 5 labels), and Patents (3,774,768 vertices, 16,518,947 edges, 20 labels).
The first three graphs are labeled real-world graphs and popular among studies of subgraph matching \cite{Han2019, Sun2020tkde, Sun2020sigmod, Sun2021, Kim2022}.
\revision{R2}{O4}{explanation}{
The last one, Patents, is the largest graph used in the recent studies \cite{Sun2020tkde, Sun2020sigmod}.
}
This is an unlabeled graph, and thus we gave the randomly-assigned labels used in the evaluation by Sun et al.~\cite{Sun2020sigmod}, which is publicly available\footnote{\url{https://github.com/RapidsAtHKUST/SubgraphMatching\#experiment-datasets}}.
We generated query graphs also in the same manner as Sun et al.; specifically, we performed a random walk on a data graph and extracted a subgraph induced by the visited vertices as a query graph.
A query graph is classified as a sparse query graph if its average degree is less than three; otherwise, it is classified as a dense query graph.
We generated query sets of sparse and dense query graphs by changing the number of vertices.
Query sets of sparse query graphs are 8S, 16S, 24S, and 32S, and those of dense query graphs are 8D, 16D, 24D, and 32D.
Thus, there are 32 query sets in total for four data graphs, four sizes, and two densities.
\revisionx{R1}{O2}{50k-text}{
Each query set contains 50,000 query graphs.
While it is popular to make a query set of 100 or 200 query graphs \cite{Bi2016, Han2019, Sun2020sigmod, Sun2020tkde, Sun2021, Kim2022}, it is too few considering that an $n$-vertex query graph has $(n - 1)!$ possible topologies and $|\Sigma|^n$ possible label assignments.
Although certain applications such as crime detection \cite{Michalak2011, Qiu2018} focus on subgraphs that rarely occur in a data graph, they tend not to be extracted as a query graph.
Thus, large query sets are necessary to extensively evaluate the efficiency of each method.
}

\emph{Machine and terminate conditions.}
We conducted the experiments on a machine with four Intel Xeon E7-8890 v3 processors (18 cores per socket, and thus 72 cores in total) and 2 TB of memory.
Except for the evaluation of parallelism (\cref{sec:eval-parallel}), all the methods were executed in a single thread using one physical core exclusively.
To reduce the experimental time, we used up to 70 cores to run 70 experiments simultaneously.
Similarly to the existing studies \cite{Bi2016, Han2019, Sun2020tkde, Sun2020sigmod, Kim2022}, we terminated the search for a query graph when $10^5$ embeddings were discovered.
We set a time limit for a query graph and a query set, respectively.
A search for a single query graph was terminated after one hour.
On the other hand, the query set was divided into subgroups of 100 query graphs, and when the total processing time of any subgroup exceeded three hours, the whole query set was judged as a ``did not finish'' (DNF).

\subsection{Comparison with Existing Methods}

We first focus on the distribution of the processing time of each query and then show the average time.
We consider the distribution more informative because the average is largely affected by the setting of the time limit; specifically, a short time limit hides the impact of expensive query graphs, and in contrast, a long time limit lets expensive query graphs dominate the result.

\subsubsection{Distribution of Processing Time}

\begin{table}[!t]
  \setlength\tabcolsep{0.4mm}
  \footnotesize
  \centering
  \caption{\revision{R1}{O3}{dnf}{Finished (i.e., non-DNF) query sets}}
  \vspace{-3mm}
  \label{tb:dnf-query-sets}
  \begin{tabular}{llcccccccclcccccccclcccccccclr}
    \bhline{1pt}
 & & \multicolumn{8}{c}{Human} & & \multicolumn{8}{c}{WordNet} & & \multicolumn{8}{c}{Patents} & & \multirow{2}{*}{\rotatebox[origin=r]{90}{Count}} \\ \cline{3-10}\cline{12-19}\cline{21-28}
      & & \rotatebox[origin=r]{90}{8S} & \rotatebox[origin=r]{90}{16S} & \rotatebox[origin=r]{90}{24S} & \rotatebox[origin=r]{90}{32S} & \rotatebox[origin=r]{90}{8D} & \rotatebox[origin=r]{90}{16D} & \rotatebox[origin=r]{90}{24D} & \rotatebox[origin=r]{90}{\ 32D} & & \rotatebox[origin=r]{90}{8S} & \rotatebox[origin=r]{90}{16S} & \rotatebox[origin=r]{90}{24S} & \rotatebox[origin=r]{90}{32S} & \rotatebox[origin=r]{90}{8D} & \rotatebox[origin=r]{90}{16D} & \rotatebox[origin=r]{90}{24D} & \rotatebox[origin=r]{90}{32D} & & \rotatebox[origin=r]{90}{8S} & \rotatebox[origin=r]{90}{16S} & \rotatebox[origin=r]{90}{24S} & \rotatebox[origin=r]{90}{32S} & \rotatebox[origin=r]{90}{8D} & \rotatebox[origin=r]{90}{16D} & \rotatebox[origin=r]{90}{24D} & \rotatebox[origin=r]{90}{32D} & & \\
    \hline
\rule{0em}{2.6ex}GuP   & & \checkmark & \checkmark & \checkmark & \checkmark & \checkmark & \checkmark &            &            & & \checkmark & \checkmark & \checkmark & \checkmark & \checkmark & \checkmark &            &            & & \checkmark & \checkmark & \checkmark & \checkmark & \checkmark & \checkmark & \checkmark & \checkmark & & 20 \\
DAF   & & \checkmark &            &            &            & \checkmark &            &            &            & & \checkmark &            &            &            &            &            &            &            & & \checkmark & \checkmark &            &            & \checkmark & \checkmark & \checkmark &            & &  8 \\
GQL-G & & \checkmark & \checkmark &            & \checkmark & \checkmark & \checkmark &            &            & & \checkmark & \checkmark & \checkmark &            & \checkmark & \checkmark &            &            & & \checkmark & \checkmark & \checkmark & \checkmark & \checkmark & \checkmark & \checkmark &            & & 17 \\
GQL-R & & \checkmark & \checkmark &            &            & \checkmark &            &            &            & & \checkmark & \checkmark & \checkmark & \checkmark & \checkmark & \checkmark &            &            & & \checkmark & \checkmark & \checkmark & \checkmark & \checkmark & \checkmark & \checkmark &            & & 16 \\
RM    & & \checkmark &            &            &            & \checkmark &            &            &            & & \checkmark & \checkmark & \checkmark &            & \checkmark & \checkmark &            &            & & \checkmark & \checkmark & \checkmark & \checkmark & \checkmark & \checkmark & \checkmark &            & & 14 \\
    \hline
  \end{tabular}
\end{table}

\begin{figure}[!t]
  \centering
  \includegraphics[width=1.0\linewidth]{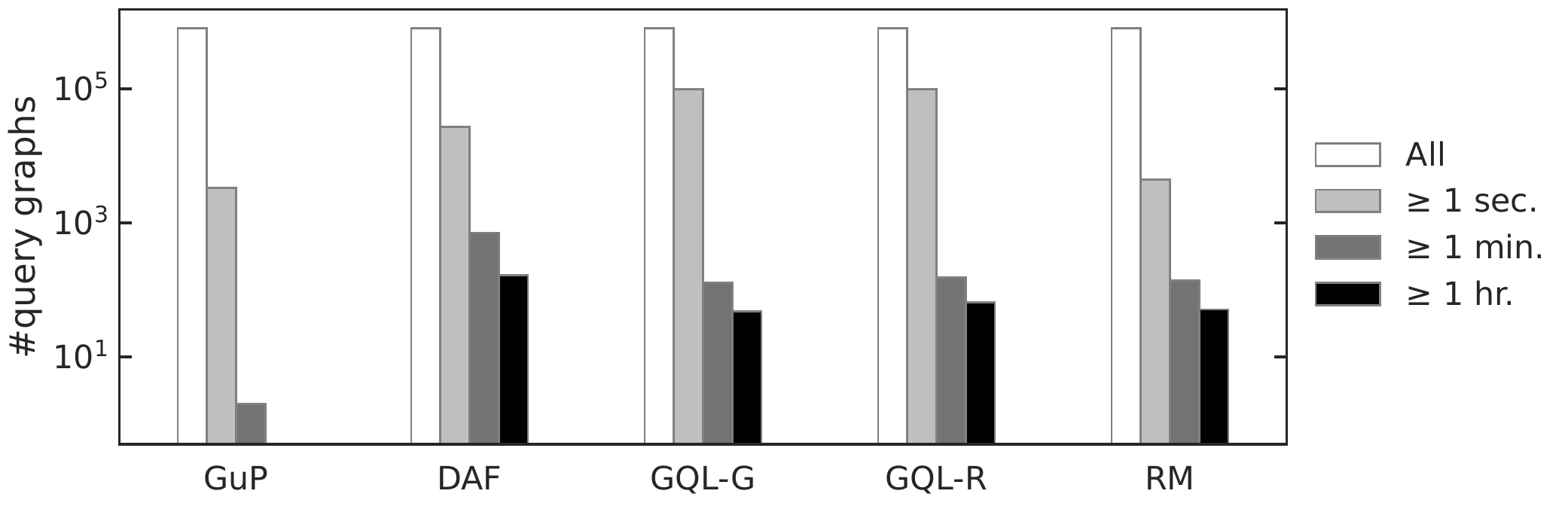}
  \vspace{-5mm}
  \caption{\revision{R1}{O2}{timeout-result}{Total number of query graphs in each processing time range.}\revisiontag{R1}{O3}{whole}}
  \label{fig:result-time-dist-whole}
\end{figure}

\cref{tb:dnf-query-sets} shows query sets that each method finished, namely, processed all the query graphs avoiding a DNF.
GuP finished the most query sets.
In addition, GuP is the only method that could finish 24S of Human and 32D of Patents.
\revision{R1}{O2}{timeout-text}{
\cref{fig:result-time-dist-whole} shows the processing time distribution of query graphs.
To equalize the number of query graphs for all the methods, we focused on 16 query sets for which no method yielded a DNF.
The ``All'' bar indicates 800,000 ($50,000 \times 16$) query graphs in those query sets.
We counted the number of query graphs that took a processing time more than the following thresholds: one second (> 1 sec.), one minute (> 1 min.), and one hour (> 1 hr.).
}
Note that, since the time limit per query graph is set to one hour, all the query graphs that took more than an hour were terminated before their completion.
As shown in the figure, GuP yielded the fewest query graphs for all the thresholds.
Most notably, GuP has no query graphs that took more than an hour.
The overall results in \cref{tb:dnf-query-sets,fig:result-time-dist-whole} confirm the high robustness of GuP, which enables GuP to process query graphs in a practical time that the state-of-the-art methods cannot.

Next, we present the processing time distribution for query sets 16S, 32S, 16D, and 24D of each data graph.
Like \cref{fig:result-time-dist-whole}, \cref{fig:overall-result} shows bars of the number of query graphs that took more than a second, a minute, and an hour.
Instead of the ``All'' bars, the top of the Y axis is set to 50,000, the number of query graphs in each query set.
GQL-G and GQL-R are shown as ``G.-G'' and ``G.-R'' because of space limitation.
GuP showed shorter query processing time than the existing methods as a whole, and we can confirm the stable performance of GuP for various query graphs and data graphs.
In addition, GuP always yielded the fewest query graphs that took more than an hour, except for 16D of WordNet.
This proves that GuP can effectively reduce the search space of difficult queries.

\revision{R1}{O2}{50k-result}{
  \cref{fig:overall-result} shows that in many cases the number of query graphs that took over an hour was less than 100, which is 0.2\% of 50,000 query graphs in each query set.
  They would have not been found if each query set had consisted of 100 or 200 query graphs.
}

\begin{figure}[t]
  \begin{minipage}[b]{\linewidth}
    \centering
    \includegraphics[width=1.0\linewidth]{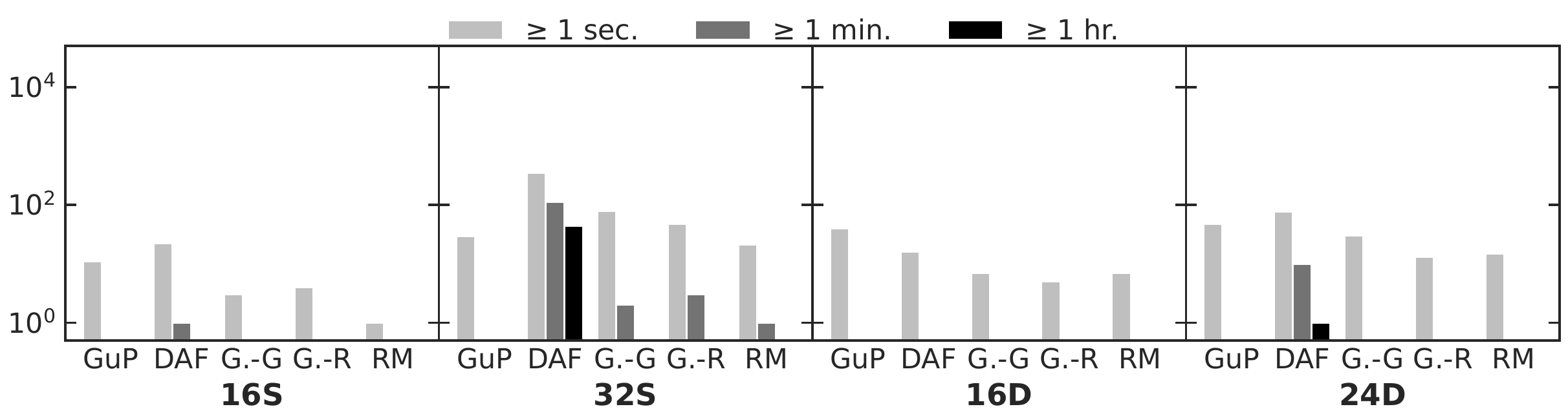}
    \vspace{-5ex}
    \subcaption{Yeast}
    \label{fig:overall-yeast-result}
  \end{minipage}
  \begin{minipage}[b]{\linewidth}
    \centering
    \includegraphics[width=1.0\linewidth]{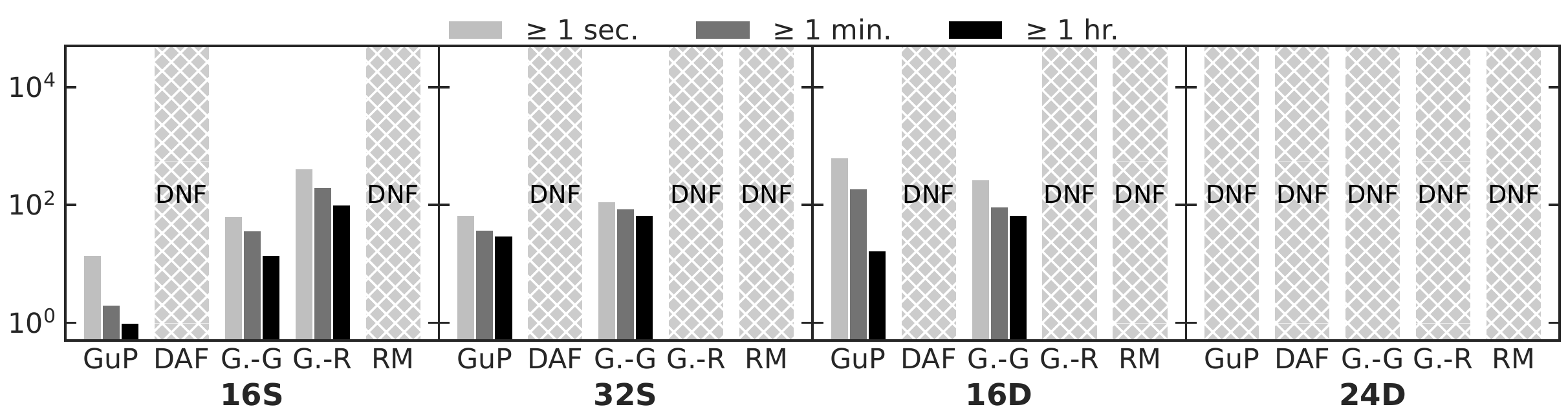}
    \vspace{-5ex}
    \subcaption{Human}
    \label{fig:overall-human-result}
  \end{minipage}
  \begin{minipage}[b]{\linewidth}
    \centering
    \includegraphics[width=1.0\linewidth]{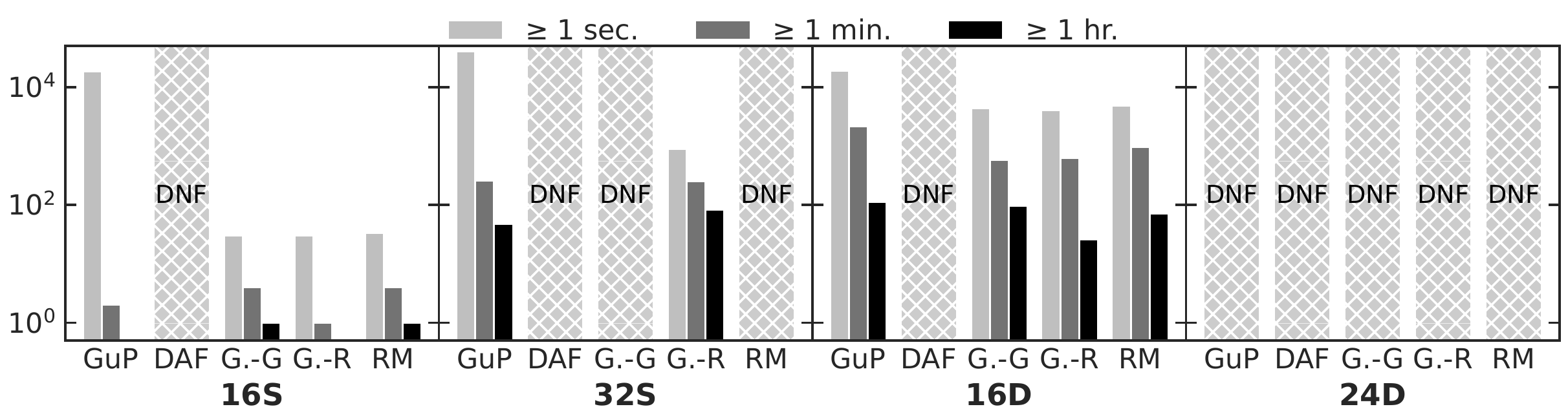}
    \vspace{-5ex}
    \subcaption{WordNet}
    \label{fig:overall-wordnet-result}
  \end{minipage}
  \begin{minipage}[b]{\linewidth}
    \centering
    \includegraphics[width=1.0\linewidth]{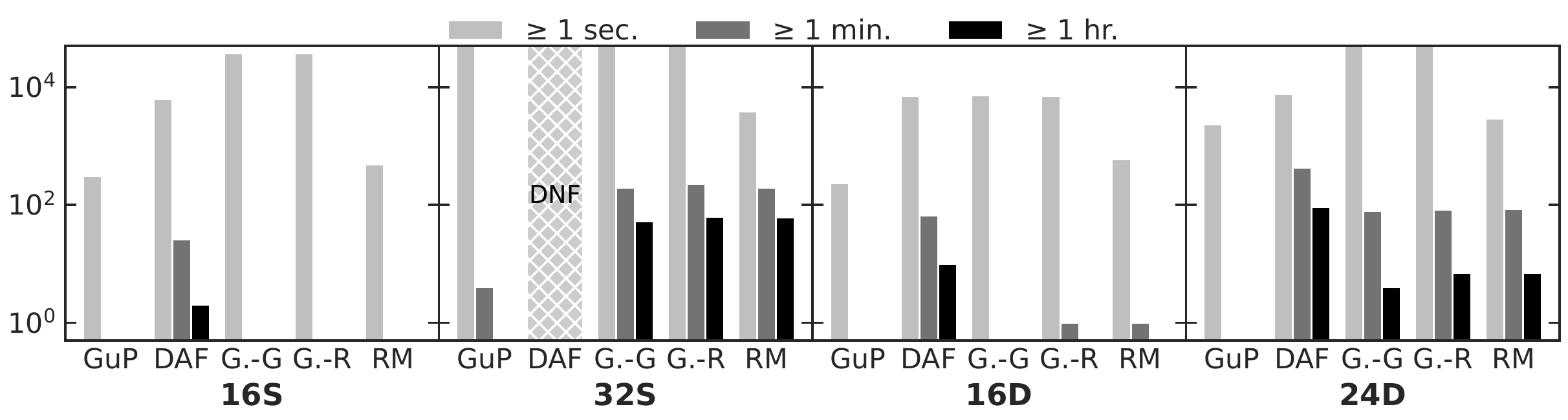}
    \vspace{-5ex}
    \subcaption{Patents}
    \label{fig:overall-patents-result}
  \end{minipage}
  \vspace{-6mm}
  \caption{\revision{R1}{O3}{breakdown}{Breakdown of the number of query graphs.}}
  \label{fig:overall-result}
\end{figure}

\subsubsection{Average Processing Time}

\begin{figure}[t]
  \includegraphics[width=\linewidth]{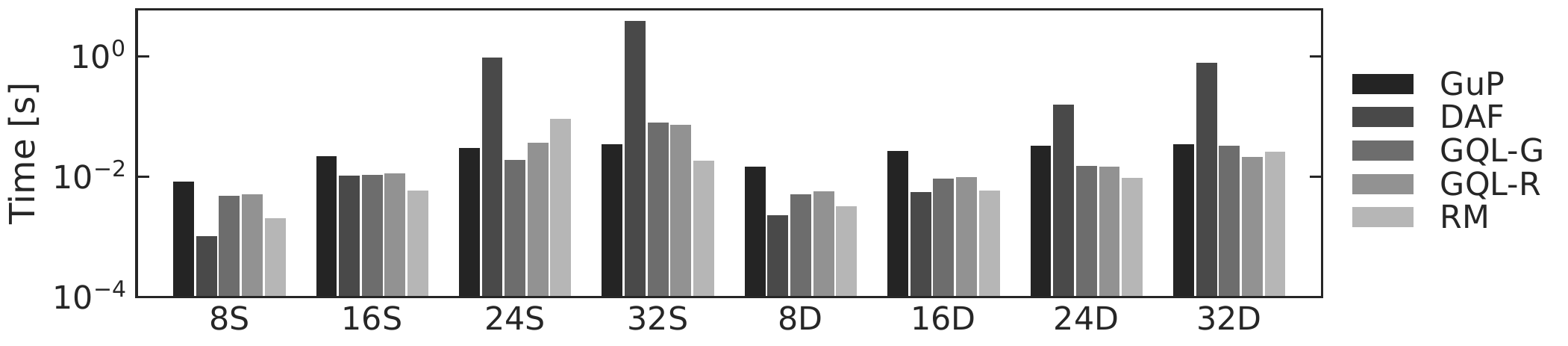}
  \vspace{-6mm}
  \caption{Average processing time for each query set of Yeast.}
  \label{fig:search-sec-yeast}
\end{figure}

\cref{fig:search-sec-yeast} presents the average processing time of each query graph in the query sets of Yeast.
The results for the other data graphs are omitted because of a lack of comparability caused by DNF query sets.
Timed-out query graphs are counted as if they were completed in one hour, which is the time limit per query graph.
\revision{R1}{O4}{}{
This figure reveals another aspect of the performance because \cref{fig:result-time-dist-whole,fig:overall-result} classify query graphs into the ranges of the processing time and do not care an actual value of the processing time in each range.
Since the generation of and the matching with guards involve additional overheads, GuP yielded only moderate performance for 8- and 16-vertex query graphs.
}
However, GuP became one of the best methods for 24- and 32-vertex query graphs because larger query graphs have larger search space, where the performance gain offered by guards more easily surpasses the overheads.
As we can confirm from the processing time distribution and the number of DNFs depicted in \cref{fig:overall-result}, the query graphs of the other data graphs are even more difficult to solve, and hence GuP tends to perform better than the other methods.

\revisiontext{
\subsubsection{Number of Recursions}

\begin{figure}[t]
  \begin{minipage}[b]{0.48\linewidth}
    \centering
    \includegraphics[width=1.0\linewidth]{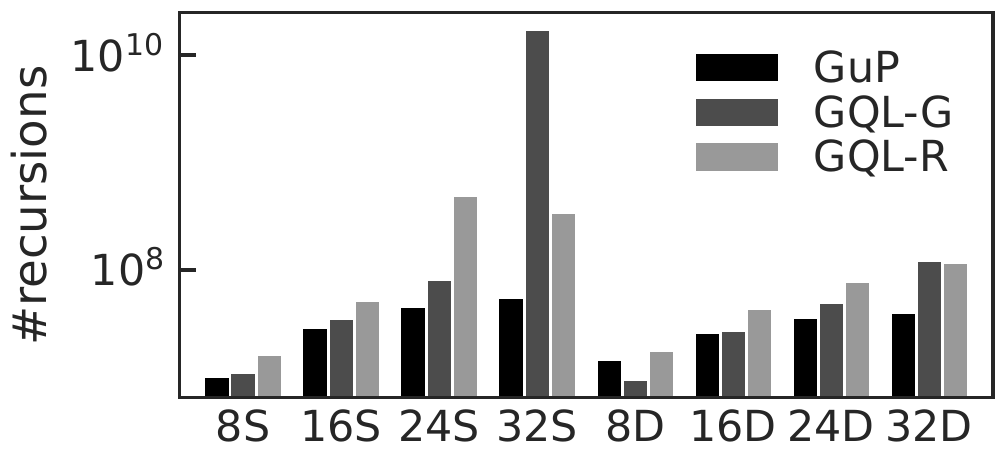}
    \vspace{-5mm}
    \caption{
      \revisiontext{
        Comparison of the number of
        \revisiontaghere{R2}{O3}{result}
        recursions.
        \revisiontaghere{R3}{O2}{result}
      }
    }
    \label{fig:rec-count-comparison}
  \end{minipage}
  \hspace{0.01\linewidth}
  \begin{minipage}[b]{0.48\linewidth}
    \centering
    \includegraphics[width=1.0\linewidth]{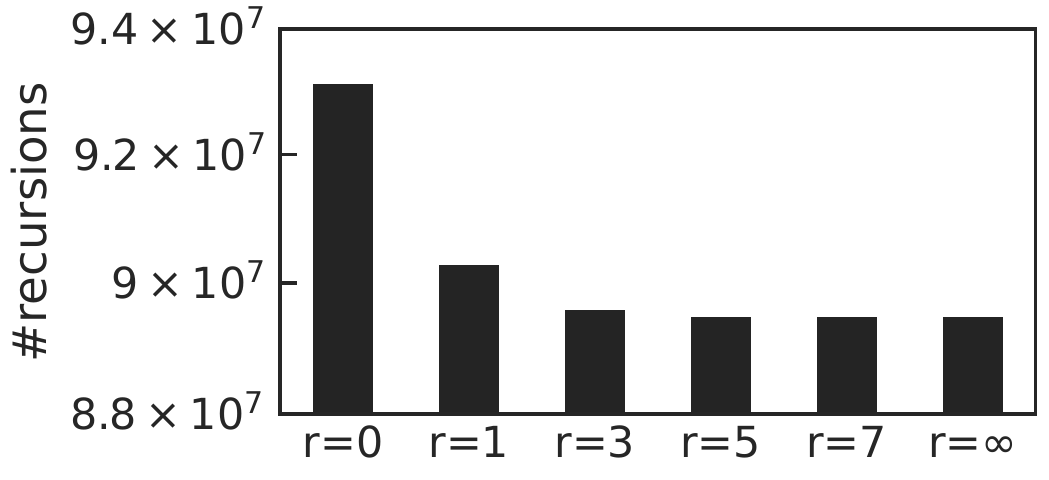}
    \vspace{-5mm}
    \caption{
      \revisiontext{ 
        Parameter search for reservation size $r$.
      }
    }
    \label{fig:resv-size-rec-count}
  \end{minipage}
\end{figure}

\revisiontag{R1}{D1}{result-text}
To evaluate the size of the search space, we compared the number of recursive calls of the backtracking function.
\revisiontag{R2}{O3}{text}
\cref{fig:rec-count-comparison} shows the total number of recursions needed to process each query set of Yeast.
\revisiontag{R3}{O2}{text}
We omitted DAF and RM because they do not count the recursions; DAF employs leaf decomposition \cite{Bi2016} besides backtracking, and RM is a join-based method.
As shown in the figure, GuP produced the fewest number of recursions for most of the query sets.
This result shows that the high performance of GuP is derived from the reduction of the search space.
Note that, due to overheads related to guards, GuP showed longer average processing time in \cref{fig:search-sec-yeast} contrary to fewer recursions.

\revision{R2}{O5}{}{
We also counted the number of local candidate vertices adaptively pruned by guards during backtracking.
While we omit the detailed results, 11.5\% of local candidate vertices were pruned on average.
This may seem a slight reduction but greatly impacts the number of recursions because it is determined by the multiplication of the number of local candidate vertices.
For example, if we have a 32-vertex query graph, and 11.5\% of local candidate vertices are pruned for every query vertex, the number of recursions decreases to 2\% ($(1 - 0.115)^{32} = 0.02$).
}
}

\subsection{Detailed Analysis of GuP}

Next, we show the results of the experiments to understand the characteristics of GuP.

\revisiontext{
\subsubsection{Reservation Size}
\label{sec:parameter-r}

GuP requires parameter $r$, which specifies the maximum size of reservation guards. 
\cref{fig:resv-size-rec-count} shows the total number of recursions needed to solve 1,000 queries in each query set of Yeast.
Each bar corresponds to a different value of $r$.
``r = $\infty$'' has no limitation on the size.
We disabled the pruning techniques except for reservation guards.
The results show that the pruning power of reservation guards increases as $r$ increases, but it almost saturates at $r = 3$.
We also confirmed almost the same trends with the other data graphs.
From this result, we recommend $r = 3$ as the default setting because $r$ is preferred to be small to reduce the computational costs of the reservation guard generation and a matching test with reservation guards.
Remind that we always used $r = 3$ except for this experiment.
}

\subsubsection{Effectiveness of Each Guard}

\begin{figure}[!t]
  \centering
  \includegraphics[width=1.0\linewidth]{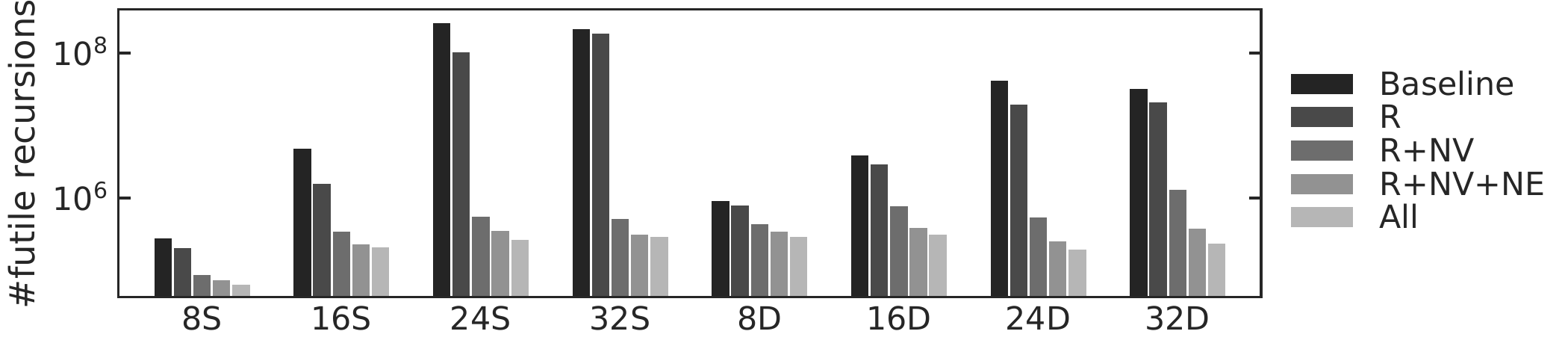}
  \vspace{-4mm}
  \caption{The number of futile recursions on Yeast.}
  \label{fig:recursion-reduction}
\end{figure}

Next, we investigated the effectiveness of each guard.
To better understand the contribution of guards, here we focused on the number of \emph{futile} recursions that is a recursive call leading to a deadend.
\cref{fig:recursion-reduction} shows the number of futile recursions offered by the different combinations of techniques in GuP.
``Baseline'' means a conventional backtracking search, ``R'', ``NV'', and ``NE'' mean the use of reservation guards, nogood guards on vertices, and nogood guards on edges, respectively.
Finally, ``All'' means complete GuP, equivalent to ``R+NV+NE'' with backjumping.
We can see the overall trend that nogood guards on vertices (``NV'') contributed the most to the reduction of futile recursions.
The contribution of nogood guards on edges (``NE'') is the second largest, and backjumping (``All'') offers a little bit more improvement.
Although the contribution of reservation guards (``R'') varied among the query sets, they substantially decreased the number of futile recursions for 16S, 24S, and 24D by 77\%, 60\%, and 53\%, respectively.
Thus, all the techniques in GuP contribute to achieving efficient backtracking, leading to the high robustness of GuP.

\revisiontext{
\subsubsection{Memory Consumption}

\begin{table}[!t]
  \setlength\tabcolsep{1mm}
  \footnotesize
  \centering
  \caption{\revision{R2}{O4}{memory-result}{Peak\revisiontag{R3}{O2}{memory-result} memory consumption}}
  \vspace{-3mm}
  \label{tb:result-memory}
  \begin{tabular}{lrlrrrrrr}
    \bhline{1pt}
                             &       & &                           & \multicolumn{3}{c}{Guard}            &             \\
    \cline{5-7}
    Graph                    & Query & & \multicolumn{1}{c}{Whole} & Reservation & N. vertices & N. edges & Guard/Whole \\
    \hline
    \multirow{4}{*}{Yeast}   &    8S & &                   2.91 MB &     0.07 MB &     0.09 MB &  0.61 MB &     26.49\% \\
                             &   32S & &                   4.21 MB &     0.11 MB &     0.13 MB &  0.83 MB &     25.36\% \\
                             &    8D & &                   3.37 MB &     0.04 MB &     0.05 MB &  0.81 MB &     26.59\% \\
                             &   32D & &                   4.27 MB &     0.07 MB &     0.08 MB &  1.00 MB &     26.93\% \\
    \hline
    \multirow{4}{*}{Patents} &    8S & &                   1.51 GB &     1.43 MB &     1.86 MB &  2.14 MB &      0.36\% \\
                             &   32S & &                   1.51 GB &     2.53 MB &     3.09 MB &  4.31 MB &      0.66\% \\
                             &    8D & &                   1.51 GB &     0.50 MB &     0.67 MB &  0.68 MB &      0.12\% \\
                             &   32D & &                   1.51 GB &     1.39 MB &     1.66 MB &  3.04 MB &      0.40\% \\
    \hline
  \end{tabular} 
\end{table}

\revisiontag{R2}{O4}{memory-text}
Since
\revisiontag{R3}{O2}{memory-text}
GuP needs additional memory space for guards, we evaluated its memory consumption using Yeast and Patents, the largest data graph in our experiment.
\cref{tb:result-memory} shows the result.
The ``Whole'' column shows the peak heap memory consumption\footnote{We used heaptrack to obtain these values: \url{https://github.com/KDE/heaptrack}.}, the columns under ``Guard'' shows the maximum memory consumption of each guard, and ``Guard/Whole'' shows the percentage of the total memory consumption of guards in the whole memory consumption.
While guards occupied about one fourth of the whole memory consumption for Yeast, the percentage decreased to under 1\% for Patents.
This is because the memory consumption for Patents is dominated by the data graph.
The program needs much temporary memory for buffering data read from files and constructing a data structure of the data graph.
In contrast, guards consume little memory because they are attached to candidate vertices and edges, which are much fewer than the vertices and edges of the data graph.
Guards are generated after releasing memory for the temporary data, and hence the peak memory consumption for Patents was 1.51 GB regardless of the size of query graphs.
Note that this seems a reasonable memory consumption because we observed that GQL-G and GQL-R also allocated about 1.5 GB of memory for Patents.
As shown by these results, guard-based pruning is applicable to large-scale graphs.
}

\revisiontext{
\subsubsection{Parallelization}
\label{sec:eval-parallel}

\begin{figure}[t]
  \begin{minipage}[b]{0.48\linewidth}
    \centering
    \includegraphics[width=1.0\linewidth]{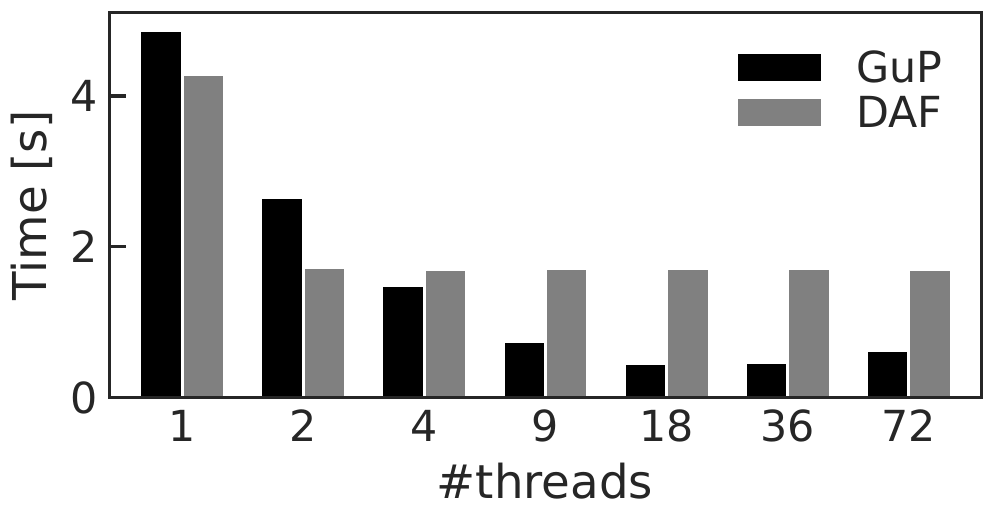}
     \vspace{-4mm}
    \subcaption{Average processing time}
    \label{fig:result-par-avg}
  \end{minipage}
  \hspace{0.01\linewidth}
  \begin{minipage}[b]{0.48\linewidth}
    \centering
    \includegraphics[width=1.0\linewidth]{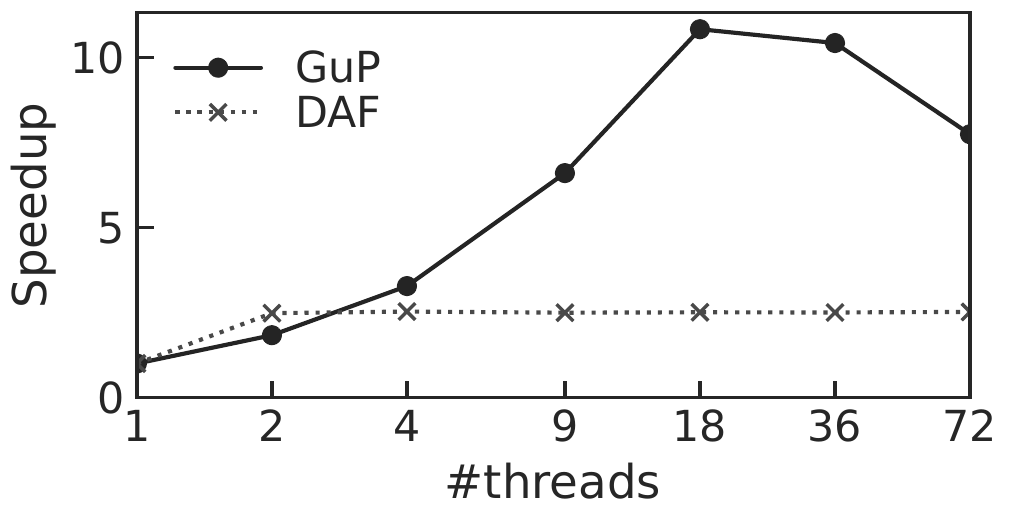}
     \vspace{-4mm}
    \subcaption{Speedup}
    \label{fig:result-par-speedup}
  \end{minipage}
  \label{fig:result-par}
  \vspace{-6mm}
  \caption{\revision{R3}{O3}{eval-result}{Performance in parallel executions.}}
\end{figure}

\revisiontag{R3}{O3}{eval-text}
We compared the performance of GuP and DAF in parallel execution because DAF is the only parallelized method among the methods used in the experiment.
Parallel search often offers superlinear speedups \cite{Han2019} when a thread encounters a search space where it can easily find embeddings more than the limit, which is $10^5$ in our experiment.
This is beneficial in practice but becomes noise in a study of parallel scalability.
To mitigate this effect, we increased the limit to $10^8$ in this experiment.
\cref{fig:result-par} shows average processing time and speedup for 1,000 query graphs in 32D of Yeast with different numbers of threads.
The 1-thread performance differs from \cref{fig:search-sec-yeast} because of the different limit on the number of embeddings.
For 1- and 2-thread execution, GuP performed worse than DAF due to guard overheads and superlinear speedup of DAF.
However, the performance of DAF does not scale to more than two threads.
This is because DAF parallelizes the search only at the candidate vertices of $u_0$ \cite{Han2019} and thus failed in load balancing.
In contrast, thanks to work stealing, GuP offered speedup almost in proportion to the number of threads and outperformed DAF with threads more than two.
Since our machine consists of four NUMA nodes each of which has 18 cores, communication costs degraded the 36- and 72-thread performances.
NUMA optimizations will improve the performance, but it is not a focus of this paper.

In the parallel execution, each thread of GuP individually maintains nogood guards and does not share them with the other threads.
Since this may affect the performance, we counted the total number of recursions in parallel execution.
Perhaps counterintuitively, the parallel execution decreased the number of recursions; the 1- and 72-thread executions produced 38.6 billion and 38.5 billion recursions, respectively.
As mentioned above, a parallel search can find search space that is easy to find many embeddings, which leads to fewer recursions.
Compared to this phenomenon, the thread-local maintenance of nogood guards has only an unobservable impact, and so pruning with guards can be applicable to parallel search.
}

\section{Conclusion}
\label{sec:conclusion}

We proposed GuP, an efficient algorithm for subgraph matching.
GuP utilizes guards on candidate vertices and candidate edges to filter out them adaptively to partial embeddings.
Our contributions are (i) a pruning approach based on guards, (ii) the propagation of the injectivity constraint by a reservation, (iii) the nogood discovery rules for effective pruning, and (iv) search-node encoding of a nogood guard.
The experimental results showed that GuP can solve many queries that the state-of-the-art methods could not solve within a time limit and also can solve other queries in comparable processing time.

\bibliographystyle{ACM-Reference-Format}
\bibliography{extracted}

\end{document}